\documentclass[10pt]{amsart}
\usepackage{amsmath}
\usepackage{amsfonts}
\usepackage{amssymb}
\usepackage{latexsym}
\usepackage{mathrsfs}

 \newtheorem{prop}{Proposition}
 \newtheorem{theo}{Theorem}
 \newtheorem{lemma}{Lemma}
 \newtheorem{corol}{Corollary}

\theoremstyle{remark}
\newtheorem{rem}{Remark}

\def\om{\omega}

\def\de{\delta}
\def\si{\sigma}
\def\var{\varepsilon}

\def\R{{\Bbb R}}
\def\E{{\Bbb E}}
\def\N{{\Bbb N}}

\def\pa{\partial}

\def\Cal{\mathcal}
\def\Scr{\mathscr}
\def\cO{{\Cal O}}
\def\cP{{\Cal P}}

\def\Som{\sum\limits}
\def\Sup{\sup\limits}
\def\Lim{\lim\limits}
\def\Bigcup{\bigcup\limits}
\def\lmul{\substack}
\def\deb{\rightharpoonup}
\def\blackbox{\unskip\kern 6pt\penalty 500%
\raise -1pt\hbox{\vrule\vbox to 8pt{\hrule width 6pt\vfill\hrule}\vrule}}

\title[particle model for coupled PDE's] {A simple particle model for a system of coupled equations with absorbing collision term}

\author[C\'edric Bernardin and Valeria Ricci]{}


\email{Cedric.Bernardin@umpa.ens-lyon.fr}
\email{valeria.ricci@unipa.it}
\thanks{Valeria Ricci acknowledges the support by the GNFM (research project 2008:``Approssimazione di equazioni alle derivate parziali tramite sistemi di particelle").}

\begin{document}
\maketitle

\centerline{\scshape C\'edric Bernardin}
\medskip
{\footnotesize
 \centerline{Universit\'e de Lyon and CNRS,}
   \centerline{UMPA, UMR-CNRS 5669, ENS-Lyon,}
   \centerline{46, all\'ee d'Italie, 69364 Lyon Cedex 07 - France}
}

\medskip

\centerline{\scshape Valeria Ricci}
\medskip
{\footnotesize
 \centerline{Dipartimento di Metodi e Modelli Matematici}
   \centerline{Universit\`a di Palermo, Viale delle Scienze Edificio 8}
   \centerline{I90128 Palermo - Italy}
}


\begin{abstract}
We study a particle model for a simple system of partial differential equations describing, in dimension $d\geq 2$, a two component mixture  where light particles move in a medium of absorbing, fixed obstacles; the system consists in a transport and a reaction equation coupled through pure absorption collision terms. We consider a particle system where the obstacles, of radius $\var$, become inactive at a rate related to the number of light particles travelling in their range of influence at a given time and the light particles are instantaneously absorbed at the first time they meet the physical boundary of an obstacle; elements belonging to the same species do not interact among themselves. We prove the convergence (a.s. w.r.t. the product measure associated to the initial datum for the light particle component) of the densities describing the particle system to the solution of the system of partial differential equations in the asymptotics  $ a_n^d n^{-\kappa}\to 0$ and $a_n^d \var^{\zeta}\to 0$, for $\kappa\in(0,\frac 12)$ and $\zeta\in (0,\frac12 - \frac 1{2d})$, where $a_n^{-1}$ is the effective range of the obstacles and $n$ is the total number of light particles.
\end{abstract}

\section{Introduction}

In this paper we propose a microscopic model for a  system of partial differential equations consisting in a transport equation, having pure loss collision term, and a pure loss reaction equation; the equations are  self--consistently coupled and they are meant to describe the evolution in time of the phase space densities in a binary mixture where the two species interact in such a way that each species inhibits the activity of the other, with an interaction proportional to their macroscopic (position space) densities. In particular, we want to rigorously derive the system of equations in a suitable asymptotics for the particle system.

The system of equations we shall deal with is a sort of very simple reactive transport system, though without conservation of masses; more complex reactive transport systems arise in many different contexts, like for instance in the modeling of biological systems, of porous media, of radiative transfer or, more in general, of various systems in the presence of chemical reactions (see e.g.\cite{SDL,MM}).
We are interested here in the analysis of the simplest reactive self-consistent coupling (i.e. the absorption coupling), and we shall therefore not include other terms in the equations.

From the point of view of particle modeling, an analysis of the absorption self-consistent coupling for one species nonlinear equations
can be found in
models for reaction-diffusion equations, such as the ones proposed in \cite{NOR} and \cite{Szn}; in both papers, the particle system evolves according to a Brownian motion (to get the diffusion) and destruction of particles occurs with different mechanism: in the first paper the destruction is stochastic, with a death rate for the Brownian particles which is a function of stochastic exponential times, in the second one it is deterministic, and it happens as soon as collisions between particles occur.

We shall consider a binary semi--deterministic system where both kinds of interactions occur.
More precisely, the particle system we shall consider consists in point like (light) particles  moving uniformly among fixed, spherical particles (obstacles). Particles not belonging to the same species do  not interact among themselves and particles of different species interact in the following way: a light particle becomes inactive (or is absorbed) at the first time it meets an obstacle and an obstacle becomes inactive (or disappears) in a stochastic interval of time whose size is connected, through a local mean-field type interaction, to the number of light particles traveling within the area of detection (range) of the obstacle.

The main difficulty in the derivation of the (otherwise simple) system of partial differential equations originates from the self consistent structure of the problem: in order to overcome the mathematical troubles introduced from the self--consistent terms, we shall adopt a natural scheme for facing self--consistent problems in partial differential equations and particle systems and we shall prove the convergence using two levels of approximation of the original particle system, the first one eliminating the self-consistent structure and the second one reducing the correlations between the two species with respect to the original many-particle system. Our formalism is quite explicit and the strategy we shall adopt is similar to the procedure used in \cite{NOR}. The adjustment of this procedure to our particle system is not trivial, since, at variance with the situation treated in \cite{NOR}, we deal with a two species system interacting in a non symmetric way and the instantaneous absorption of the light particles is in some sense a singular interaction with respect to the local mean-field type interaction which govern the obstacles lifetime,
i.e. the interaction type considered in \cite{NOR}. The presence of the deterministic component (the light particles) imposes a sufficiently careful analysis of the trajectories, in the style of the analysis performed for particle models of linear equations based on similar deterministic components (see e.g. \cite{BGW}, \cite{DR1, DR2}).  Nevertheless, the structure of the stochastic component (the obstacles)  simplifies by a considerable amount the mathematics with respect to a two component, totally deterministic system.

Our final theorem establishes a weak law of large numbers for
the empirical measures (associated to the two species of particles in the mixture) to the solution of the system of partial differential equations, almost surely with respect to the initial distribution of positions in the phase space of  the light particles and in probability with respect to the distributions of positions and life time of the obstacles. This weak law of large numbers is valid in an asymptotics where the radius and the effective range of the obstacles vanish, keeping a finite action in the limit, and the number of  particles grows up to infinity, these quantities being related in a way that  we shall determine while proving the theorem.
As a part of this relationship, we shall prove a simple condition (relating the diminishing rate of the effective range of the obstacles to the number of light particles)  which guarantees the weak convergence of the product of an empirical measure times a sufficiently regular approximant of a Dirac delta distribution toward the product of their weak limits (under suitable regularity assumptions on the weak limit of the empirical measure and on the choice of the delta's approximant). This condition entails very useful bounds for our estimates and allows to get easily the required asymptotic result.

\section{The equations, the particles model and the main result}
\label{sezdue}

Throughout the paper we shall use the following notations: in dimension $d\geq 2$, we denote by  $x\in\R^d$ the position of a light particle and by $v\in\R^d$ its velocity;
$t\in\R_+$ is the time variable. In general, configurations of $M$ variables in $\R^d$ will be denoted by boldface letters with subscript $M$ ($\mathbf{y}_M=(y_1, \ldots, y_M)$) and sequences of variables by capital letters with the subscript $\infty$ ($Y_\infty =\{y_i\}_{i=1}^{\infty}$). When needed (e.g., in the functional spaces) we shall use $\R_x^d$ and $\R_v^d$ resp. for the position and the velocity spaces.
For $p\in\R^d$ and $r\in \R_+ $ we shall denote by $B_r(p)=\{y\in \R^d : |p-y|\leq r\}$ the ball of radius $r$ centered in $p$. For a given $z=(x, v)\in \R^d\times\R^d$, the free flow associated to the light particle with initial position (in the phase space) $z$ is
$$T^t(z)=(T^t_1(z), T^t_2(z))=(x+vt, v),\qquad t\in\R_+,$$
while $x (t)=T^t_1(z)$, $t\in\R_+$, is its trajectory,
and for $\var>0$
\begin{equation}
\label{tubo}
{\Cal T}_\var (t,z)=\left\{ y\in \R^d \quad; \exists s \in [0,t), \quad |y-x(s)| \leq \varepsilon \right\}
\end{equation}
denotes the flow tube of radius $\var$ associated to the trajectory
up to the time $t$.

Unless differently stated, for a given random variable $\eta$, we denote by $P_{\eta}$ its associated probability distribution, and for a  measure $\pi$ (random variable $\eta$), we denote by $\E_{\pi}$  ($\E_{\eta}$) the expectation w.r.t. $\pi$ ($\eta$).
In order to simplify the notation, in many stochastic variables depending on configurations $\mathbf{y}_M$, we shall label this dependence by a simple subscript $M$, instead of rewriting each time the whole configuration.

We shall moreover denote $\deb$ the weak convergence (convergence in law) in the space of finite measures and $\stackrel{*}{\deb}$ the *-weak (vague) convergence on the space of Radon measures, and by $K$ (with some subscript) any generic constant whose value needs not to be specified. We shall sometimes shorten the notation for sums and difference of functions with the same argument as $(f\pm g)(w)=f(w)\pm g(w)$.

\smallskip

For $T>0$, we consider the following system of partial differential equations for the two densities $f=f(t,x,v)$, $\si=\si(t,x)$, $t\in [0,T]$, $x\in \R^d$, $v\in \R^d$
:
\begin{eqnarray}
\label{sistlim}
\left\{
\begin{array}{l}
\pa_t f +v\cdot\nabla_x f = -C_d |v|\si f \\
\pa_t \si= -\Theta (\int_{\R^d} dv f ) \si\\
f(0,x,v)=f_0(x,v)\\
\si(0,x)=\si_0(x),
\end{array}
\right.
\end{eqnarray}
where $\Theta$  and $C_d=  \int_{\{\omega\in\R^d : |\omega|=1\}} |n\cdot \omega| d\omega$, $|n|=1$, are positive constants.
In order to have a convenient existence and uniqueness theorem for the solution (cf. Appendix \ref{app1}), we assume $f_0\in L^1(\R_v^d; W^{1,\infty}(\R_x^d))$, $vf_0\in L^1 (\R_v^d; W^{1,\infty}(\R_x^d))\cap L^{\infty}(\R_x^d\times \R_v^d)$,  $v^2 f_0\in L^1 (\R_v^d; L^{\infty}(\R_x^d))$ and $\sigma_0\in W^{1,\infty}(\R_x^d)$.

We shall show that the system (\ref{sistlim}) can be derived from a semi-deterministic particle system of the kind we specified in the Introduction. The particle system will be described in next paragraphs.

\smallskip

Obviously, once established the correct asymptotics to get equation (\ref{sistlim}) from the particle system we have chosen, we shall simply write $\Lim_{n\to\infty}$ to denote a limit in this asymptotics ($n$ and $\var$ being related). This peculiar use of the notation will be consequent to the hypothesis in the environments where the notation is used.

\subsection{The particle system: initial data statistics}

We consider configurations of spherical fixed obstacles of radius $\var$ with stochastically distributed  positions at time $t=0$ and we denote by $\mathbf{c}_M= (c_1, \ldots, c_M)$  the coordinates of their centers. Obstacles may overlap (i.e., configurations such that for some $i,k$, $|c_i - c_k|\leq 2\var$ are allowed) and  $M= 0$ corresponds to absence of obstacles.
In this paper, we shall assume that obstacles positions follow a Poisson distribution with parameter $\mu_{\var}$, i.e. that the probability
distribution of finding $M$ obstacles in a bounded measurable set
$\Lambda \subset \R^d$ is
given by:
\begin{equation}\label{poisson}
P(d{\mathbf{c}_M})= e^{-\mu_{\var}[\Lambda]_L}\frac{\mu_{\var}^M}{M!} dc_1 \dots dc_M,
\end{equation}
where (here and in what follows) $[\Lambda]_L$ denotes the Lebesgue measure of the set $\Lambda$. We shall adapt the initial datum $\si_0(x)$ to this choice for the statistics of the obstacles.

We consider then $n$ point-like particles,
located initially at positions $x_1,\ldots, x_n$ and moving uniformly among the obstacles with velocities $v_1,\ldots, v_n$; we shall denote the phase space coordinate of the $i$-th light particles  as $z_i=(x_i,v_i)$. A point in the $n$-particles phase space is denoted as $\mathbf z_n= (x_1, v_1,\!\!\! \ldots, x_n, v_n)$ and a sequence of initial data is denoted as $Z_{\infty}=\{z_i\}_{i=1}^{\infty}\in(\R^d\times\R^d)^{\infty}$.

We describe both species of particles by means of their empirical measure (see e.g. \cite{G1}).

Given a $1$-particle probability density $f_0$,
we denote by $\Scr{P}$ the (infinite product) probability measure defined on the space of infinite sequences $(\R^d\times\R^d)^{\infty}$ by $f_0$, i.e. such that, for $n=1,2,\ldots$ and $A^{(i)}_x\times A^{(i)}_v \subset \R^d\times\R^d$,

\begin{equation}
\label{pgrande}
\Scr{P}(Z_{\infty}: z_i\in
A^{(i)}_x\times A^{(i)}_v , 1\leq i\leq n )\!=\!
\prod_{k=1}^{n} \int_{A^{(k)}_x\times A^{(k)}_v}f_0(x,v)dx dv.
\end{equation}

The sequence of initial empirical measures for the light particles $\{ \mu_n^0 \}_{n=1}^{\infty}$ is then such that:
\begin{equation}
\mu_n^0 (t,x,v; \mathbf{z}_n)=\frac 1n \sum_{i=1}^n \de_{z_i}(x,v) \deb f_0(x,v) \qquad \Scr{P}-a.e. \; w.r.t. Z_{\infty}.
\end{equation}

We shall require moreover some regularity condition on $f_0$.

\subsection{The particle system: dynamics}\label{dinamica}
We define now the dynamics of the system.

We consider a sequence $\mbox{\Large{$\tau$}}_{\infty}=\{\tau_1,\ldots, \tau_k, \ldots\}$ of independently distributed exponential variables
and we define the following stochastic functions:

\begin{itemize}

\item The \textit{risk function}  at a given position $c$,

\begin{equation}
\label{funris}
V_{n, \var, M}(t,c)=\frac 1n \sum_{i=1}^{n} \int_0^t ds q_n(x_i(s) -c )\xi_{n,\var,M}(s, z_i),
\end{equation}
where $q_n(x)=a_n^d q(a_n x )$, is an (at least continuous) approximant, up to a multiplicative constant, of the Dirac delta distribution, with  $q$ a non negative, radial ($q(x)=q(|x|)$) function  such that $\int_{\R^d} q(y)dy = \Theta>0$ and $\xi_{n,\var,M}$ is the stochastic function defined below.

\item The \textit{life functions}
\footnote{We use, here and later, the notation $\Bigcup_{k=1}^M B_{\var}(c_k)\eta_{n,\var, M}(s,c_k)$ to mean $\Bigcup_{\lmul{k\in \{1,\ldots,M\}:\\\eta_{n,\var, M}(s,c_k)=1}} B_{\var}(c_k)$. Although not formally correct, this notation allows us to write less cumbersome formulas.}, resp. $\xi_{n,\var,M}(t, z)$ for a light particle with initial position and velocity $(x,v)$ and $\eta_{n,\var, M}(t,c_k)$ for an obstacle centered in  $c_k\in\{c_1, \ldots, c_M\}$,

\begin{equation}
\label{sistpart}
\begin{array}{l}
\xi_{n,\var,0}(t,z)=1 \\
\\
M\geq 1\\
\\
\xi_{n,\var,M}(t,z)=I_{\{ x(s)\notin \Bigcup_{h=1}^M B_{\var}(c_h)\eta_{n,\var, M}(s,c_h)\quad\forall s \in [0,t)\}}\\
\\
\eta_{n,\var,M}(t, c_k )=I_{\{ V_{n,\var, M}(t, c_k)< \tau_k\}}.
\end{array}
\end{equation}

\item The \textit{maximal collision time} between a light particle with initial phase space position $z=(x,v)$ and an obstacle located in $c$,
\begin{equation}
\label{tcollis}
T_{z,c}=\inf_{s\in \R_{+}}\{s: \,\,|x(s)-c|\leq \var\},
\end{equation}
with no reference to the activity of both particles. Because no life functions are involved in this definition, $T_{z,c}$ can be infinite.

\end{itemize}

\smallskip

In the particle system defined through (\ref{funris}) and (\ref{sistpart}),  the obstacles become inactive at stochastically distributed times defined at a given position $c_k$ through the risk function (\ref{funris}) as the time $t_k=t_k(c_k, \tau_k)$ s.t. $V(t_k, c_k)=\tau_k$, while the light particles are absorbed at the first time they meet an active obstacle (i.e. as soon as $|x_i(T_{coll}) -c_k|=\var$ for some $k$ and $T_{coll}<t_k$; notice that $T_{coll}$ is defined w.r.t. active particles, and it is therefore a different variable w.r.t. $T_{z,c}$).

Unless we require $q$ to have compact support,  a light particle located at $x$ can affect the life function of a far  obstacle (it suffices that $q(x-c)\neq 0$); on the contrary, a given light particle interacts only with obstacles which are met by its trajectory in space.
Since,
for a given $n$, the volume including all light particle trajectories up to time $T$, $\mathcal{V}_n$, is such that
\begin{equation}
\label{bn}
\mathcal{V}_n\subset B_{\max\limits_{i=1,\ldots,n}|x_i|+\max\limits_{j=1,\ldots,n}|v_j| T +1} (0)=B_n ,
\end{equation}
and, in studying the interacting system, we do not need to consider quantities related to obstacles which can not be met from any light particles,
we may (for a given $n$) restrict the expectation value with respect to (\ref{poisson}) to the volume $\Lambda=\Lambda_n$, for a growing sequence of bounded  Lebesgue measurable sets $\{\Lambda_n\}_n$ s.t. $B_n\subset\Lambda_n$ and $\Lim_{n\to\infty} \Lambda_n= \R^d$. We denote by $\E_c^n$ the expectation value with respect to the centers  distribution (\ref{poisson}) with such a choice for $\Lambda$.
We shall denote by $\E^n= \E^n_c \E_{\mbox{\Large{$\tau$}}}= \E_{\mbox{\Large{$\tau$}}}\E^n_c$, where $\E_{\mbox{\Large{$\tau$}}}$ is the expectation value w.r.t. to the exponential times sequence $\mbox{\Large{$\tau$}}_{\infty}$, and by $P^n$ the corresponding probability distributions.

\subsection{Scaling laws: heuristics}

In order to derive the system (\ref{sistlim}) from the particle system defined above in the macroscopic limit where the radius $\var$ of the obstacles vanishes, we have to choose the scaling for densities of both species of particles in a suitable way.

As a first requirement, we want the mean free path of the light particles to be finite at the scale at which the system is observed,  in such a way to keep track, in the chosen asymptotics, of the interaction between the light particles and the obstacles. To this purpose, we fix the rate $\mu_\var$ of the Poisson process
(\ref{poisson}) to be such that:
\begin{equation}
\label{BG}
\mu_{\varepsilon} {\varepsilon^{d-1}} =\mu
>0
\end{equation}
and, since we consider a uniform initial macroscopic distribution of obstacles, in (\ref{sistlim}) we assume $\si_0(x)=\mu$.

When $V_{n, \var, M}\equiv0$ (i.e. $q\equiv0$), formula (\ref{BG}) defines the so called Boltzmann-Grad scaling, which has been analysed for linear particle systems in particular in connection with the asymptotics of the Lorentz gas and its variants (see e.g. \cite{Ga, Sp, BBS}, \cite{DR1} for the stochastic case, \cite{BGW, G2, RW} for the periodic case
and the beautiful review \cite{G3}
-focused on the periodic case-).
When $V_{n, \var, M}\not\equiv0$, (\ref{BG}) guarantees the finiteness of the mean free path for a given light particle as soon as $V_{n, \var, M}$, in
flow tubes and for a non-negligible set of  obstacles configurations, is bounded uniformly in $n$. This is true, under the conditions on $f_0$  specified at the beginning of Section \ref{sezdue} and the conditions on $q$ given in Section \ref{dinamica} (i.e. $q\in L^1(\R^d)\cap C(\R^d)$), whenever $a_n^d n^{-\frac 12}\leq const$.

A second requirement is to have, in the $\var\to 0$ asymptotics, vanishing correlations between light particles and obstacles, so as to obtain a Markovian limit (i.e. without memory effects), coherently with the structure of (\ref{sistlim}).
Since the dominant part of correlations is associated to grazing crossings of light particle trajectories, and in particular to the mean volume $V^\var_g$ spanned from grazing trajectories crossings within the effective range of an obstacle, this requirement connects the scaling
in the effective range $a_n^{-1}$ to the scaling in the mean density of the obstacles. In order to have negligible correlations, the mean number of multiple collisions per unit volume  has to vanish in the limit.
This leads to the condition $\mu_\var a_n^d V^\var_g  \to 0.$
Notice that possibly $V^\var_g =O(\var^{d-1})$, and having a better asymptotic behavior depends essentially on the regularity of the limit measure of the light particle component.

We have then a third, technical, requirement due to the form of (\ref{funris}): since eventually we want to obtain a mean-field limit, and therefore we want to be able to perform limits of mean values with respect to $\mu_n^0$ of sequences of functions converging to singular limits (such as Dirac delta distributions), we need a condition assuring the convergence (in a suitable sense) of the product of such sequences times the initial empirical measure: this is achieved if the empirical measure converges faster to its (sufficiently regular) limit than the chosen approximant concentrates, in such a way that the empirical measure is in practice equivalent to its limit density well before than the delta's approximant concentrates in its center.
We can guess roughly that this happens if the fraction of particles fluctuating around mean values in a volume corresponding to the effective range of an obstacle vanishes in the chosen asymptotics, and therefore $\frac{1}{\sqrt{n}}=o(a_n^{-d})$. We shall prove a more precise asymptotics in Appendix \ref{app2}

Given (\ref{BG}) and the just described scaling laws, for a given configuration of obstacles, the empirical measure at time $t$ for the light particle component (representing its mesoscopic density in phase space) is:
\begin{equation}
\label{empit}
\mu_n(t,x,v; \mathbf{z}_n, \textbf{c}_M)=\frac 1n \sum_{i=1}^{n}
\de_{T^t(z_i)}(x,v)\xi_{n,\var,M}(t,z_i)
\end{equation}
and the (macroscopic) density of  the obstacle component (i.e. the number of obstacles per unit volume) is expressed as:
\begin{equation}
\label{densit}
\si_n (t,x; \mathbf{z}_n, \textbf{c}_M)=\var^{d-1}  \sum_{k=1}^{M}
\de_{c_k}(x)\eta_{n,\var,M}(t, c_k ).
\end{equation}

\subsection{The limit process}

Under the conditions we assumed on $(f_0,\si_0)$, as shown in paragraph \ref{app1} in the Appendix, the problem (\ref{sistlim}) admits a unique solution $(f,\si)\in L^{\infty}([0,T]\times\R^d\times\R^d)\times L^{\infty}([0,T]\times\R^d)$, with $\int dv f \in L^{\infty}([0,T]\times\R^d)$ (and actually $\int dv f(t,\cdot,v) \in C_b(\R^d)$). This solution can be expressed in semi-explicit form as:
\begin{equation}
\label{solsemiexpl}
\begin{split}
f(t,x,v)=f_0(x-vt,v) e^{-C_d |v| \int_0^t ds \si (s, x-v(t-s))}\\
\\
\si(t,x)=\si_0 (x) e^{- \Theta \int_0^t ds \int_{\R^d} dv f(s,x,v)}.\qquad\qquad\;
\end{split}
\end{equation}

In order to be able to compare solutions of (\ref{sistlim}) with the stochastic measures (\ref{empit}) and (\ref{densit}),
we define the risk functions and the life functions associated to the limit process described by (\ref{sistlim}) as:
\begin{equation}
\label{risslo}
\begin{split}
V_L^f (t, c) = \Theta \int_0^t ds \int_{\R^d} dv f(s, c, v)\\
\\
\eta_L^f (t, c)=I_{\{ V_L(t, c)< \tau\}}\qquad\qquad\;\;\;\\
\end{split}
\end{equation}

\begin{equation}
\begin{split}
\label{csisl}
U_L^{\si} (t, z)=C_d |T^t_2(z)| \int_0^t ds \si(s,T^s_1(z))\\
\\
\xi_L^{\si}(t,z)=I_{\{ U_L(t,z)< \tau_p\}},\qquad\qquad\qquad
\end{split}
\end{equation}
where $\tau$ and $\tau_p$ are independently distributed exponential variables
and $\tau=\tau_k$ when $c=c_k$.

In this way, the semi-explicit form of the solution (\ref{solsemiexpl}) can be re-expressed as
\begin{equation}
\label{solsemialt}
\begin{split}
f(t,x,v)=f_0(x-vt,v) \E_{\tau_p}\bigl[\xi_L^{\si}(t, T^{-t}(x,v))\bigr]
\\
\\
\si(t,x)=\si_0 (x) \E_{\tau}\bigl[\eta_L^f (t, x)\bigr].\qquad\qquad\qquad\;
\end{split}
\end{equation}

In order to simplify the notation, we shall omit in what follows the dependence on the density functions ($V_L=V_L^f$, $U_L=U_L^{\si}$).

\bigskip

We want to establish a (weak) law of large numbers for the measures (\ref{empit}) and (\ref{densit}).
More precisely, we shall prove the following theorem (we recall that $\Scr{S}(\R^n)$ denotes the space of $C^{\infty}(\R^n)$ functions of rapid decay at infinity \cite{Sch, Lions}):

\begin{theo}
\label{teorema}

Consider the non-negative functions $f_0$ and $q$ and assume

\begin{itemize}
\item $f_0\in\Scr{S}(\R^d\times\R^d)\bigcap L^1(\R^d_v; W^{1,\infty}(\R^d_x))$ is a probability density such that
\begin{equation}
\label{mom0}
\int_{\R^d} dv f_0\in \Scr{S}(\R^d) \qquad with \qquad \int_{\R^d}dvf_0(0,v)>0,
\end{equation}
and
$vf_0\in L^1(\R^d_v; W^{1,\infty}(\R^d_x))$, $v^2f_0\in L^1(\R^d_v; L^{\infty}(\R^d_x))$;
\vspace{.2cm}
\item $q$ is a radial function s.t. $q\in \mathscr{S}(\R^d)$ and $\int_{\R^d} dx q(x) =\Theta>0$;
\vspace{.2cm}
\item $\{a_n\}_{n=1}^{\infty}$ is such that $a_n>0$, $\Lim_{n\to\infty} a_n= \infty$ and  there exists some $\kappa\in (0,\frac 12)$ such that
\begin{equation}
\label{conditio}
\Lim_{n\to\infty}\frac{a^d_n}{n^{\kappa}}=0 ;
\end{equation}
\item
$\{\var\}_{n=1}^{\infty}=\{\var_n\}_{n=1}^{\infty}$ s.t. $\var_n>0$  and
\begin{equation}
\label{conditiovar}
\Lim_{
n\to\infty} a_n^d \var^{\zeta} = 0
\end{equation}
for some $\zeta\in(0,\frac12 -\frac1{2d})$.

\end{itemize}

Then, $\Scr{P}$-almost everywhere w.r.t. sequences of initial data $Z_\infty $ and in probability w.r.t. $P$ and $P_\tau$, when $n\to\infty$
\begin{equation}
\label{empit2}
\mu_n(t,x,v; \mathbf{z}_n, \textbf{c}_M)=\frac 1n \sum_{i=1}^{n}
\de_{T^t(z_i)}(x,v)\xi_{n,\var,M}(t,z_i)
\deb
f(t,x,v)
\end{equation}
\begin{equation}
\label{densit2}
\si_n (t,x; \mathbf{z}_n, \textbf{c}_M)=\var^{d-1}  \sum_{i=1}^{M}
\de_{c_i}(x)\eta_{n,\var,M}(t, c_i )
\stackrel{*}{\deb}
\si(t,x),
\end{equation}
where $\xi_{n,\var,M}$, $\eta_{n,\var,M}$ are defined in paragraph \ref{dinamica} and $(f,\si)$ is the unique solution of
\begin{eqnarray}
\left\{
\begin{array}{l}
\pa_t f +v\cdot\nabla_x f = -C_d |v|\si f \\
\pa_t \si= -\Theta (\int_{\R^d} dv f ) \si\\
f(0,x,v)=f_0(x,v)\\
\si(0,x)=\mu.
\end{array}
\right.
\end{eqnarray}

\end{theo}

We shall see in paragraph \ref{app2} in the appendix that  hypothesis (\ref{conditio}), together with suitable regularity assumptions on $f_0$ (all included in our theorem), guarantees the following $\Scr{P}$-a.s. convergence
\begin{eqnarray}
\label{conditiomomj}
|v|^j \mu_n^0(x,v) q_n (x) &\deb& \Theta |v|^j f_0(x,v) \delta_0(x)\\
\nonumber\\
\nonumber
|v|^j \mu_n^0  &\deb& |v|^j f_0 \\
\nonumber\\
\nonumber
\otimes_{i=1}^{k} |v|^j \mu_n^0 &\deb& \otimes_{i=1}^{k}|v|^j f_0
\end{eqnarray}
for $j=0, 1,2,\ldots, P$, $k=1,\ldots Q$, with given positive integers $P,Q$.

\section{Definitions of the approximating systems and proof of the main theorem}

The main difficulty in studying (\ref{sistlim}) and its associated particle system (\ref{funris}),(\ref{sistpart}) is that we have to deal with a self-consistent problem.
The first step is therefore to find, both for the limit system (\ref{sistlim}) and for the particle system defined from (\ref{funris}),(\ref{sistpart}), suitable approximating systems which do not share this self-consistent structure.

We first recall here without proof, since it will be useful for the sequel and we shall use it largely in various steps of the convergence proof, lemma 3.2 in \cite{NOR}, concerning bounds of distances of stochastic variables of the form $\eta(t)=I_{\{S(t)<\tau\}}$ in terms of distances between their associated risk function $S$.

\begin{lemma}

\label{3.2}
Let $\tau>0$ be an exponentially distributed real variable and define, for  $i=1,2$, $t\in [0,T]$, $\eta_i(t)=I_{\{S_i(t)<\tau\}}$,
where $S_i, S$ are non negative, non decreasing, right continuous random functions. Then, for any $\delta>0$
\begin{equation}
\label{rischio}
|\eta_1(t)-\eta_2(t)|\leq \frac 1\delta |S_1(t)-S(t)| + \frac 1\delta |S_2(t)-S(t)|
+I_{\{|S(t)-\tau|<\delta\}}.
 \end{equation}
\end{lemma}

\begin{rem}
\label{oss}
Whenever at least one or both risk functions $S_i$, $i=1,2$ are independent of the exponential time $\tau$, the bound simplifies in the first case as \mbox{$\E_\tau|\eta_1(t)-\eta_2(t)|\leq \frac1\delta\E_\tau|S_1(t)-S_2(t)| +2\delta$},
and  in the second case as 
$\E_\tau|\eta_1(t)-\eta_2(t)|\leq |S_1(t)-S_2(t)|$
.
\end{rem}

\subsection{Approximation of the limit system}
As shown in Appendix \ref{app1}, under suitable hypothesis on $f_0$, the solution of the system of equations (\ref{sistlim}) can be obtained as the $k\to\infty$ limit of the sequence of solutions of the sequence of linear
systems defined as:
\begin{eqnarray}
\label{sistk}
f^{(0)}(t,x,v)=f_0(x-vt,v),\quad\si^{(0)}(t,x)=\mu&&\\
\nonumber\\
\nonumber
\left\{
\begin{array}{l}
\pa_t f^{(k)} +v\cdot\nabla_x f^{(k)} = - C_d |v|\si^{(k-1)} f^{(k)} \\
\pa_t \si^{(k)}= -(\Theta\int_{\R^d} dv f^{(k-1)} ) \si^{(k)}\\
f(0,x,v)=f_0(x,v)\\
\si(0,x)=\mu
\end{array}
\right.&&
k=1\ldots \,\,\,
.
\end{eqnarray}

More precisely, the sequence of semiexplicit solutions of (\ref{sistk}) is, for $k=1,2,\ldots$,
\begin{equation}
\label{semiexplk}
\begin{split}
f^{(k)}(t,x,v)=f_0(x-vt,v) e^{-C_d |v| \int_0^t ds \si^{(k-1)} (s, x-v(t-s))}\\
\\
\si^{(k)}(t,x)=\mu e^{- \Theta \int_0^t ds \int_{\R^d} dv f^{(k-1)}(s,x,v)},\qquad\qquad\qquad\;
\end{split}
\end{equation}
 and we have, when $k\to\infty$ (cf. Appendix \ref{app1})
\begin{equation}
\label{convergenze}
\begin{array}{ll}
\|f-f^{(k)}\|_{L^{\infty}([0,T]\times\R^d\times\R^d)}\to 0\, ,
&
\|\int_{\R^d} dv (f-f^{(k)})\|_{L^{\infty}([0,T]\times\R^d)}\to 0 ,
\\
\\
\|\si-\si^{(k)}\|_{L^{\infty}([0,T]\times\R^d)}\to 0 .&
\end{array}
\end{equation}

The risk functions associated to the $k$-th system (\ref{sistk})
are defined as:
\begin{eqnarray}
\label{risoko}
\bar{V}^{(k)} (t, c) &=& \Theta \int_0^t ds \int_{\R^d} dv f^{(k-1)}(s, c, v)\\
\label{risokp}
\bar{U}^{(k)} (t, z) &=& C_d |T^t_2(z)| \int_0^t ds \si^{(k-1)}(s,T^s_1(z)),
\end{eqnarray}
and, because we shall need it later (cf. the definition of the system (\ref{cap})), we define $\bar{V}^{(0)} (t, c)=0$. Their associated life functions are:
\begin{equation}
\label{vitak}
\begin{split}
\bar{\eta}^{(k)} (t, c)=I_{\{ \bar{V}^{(k)} (t, c) < \tau\}}\\
\\
\bar{\xi}^{(k)} (t,z)=I_{\{\bar{U}^{(k)}(t, z)< \tau_p\}},
\end{split}
\end{equation}
with $\tau$ and $\tau_p$ exponentially distributed times. Of course, (\ref{semiexplk}) can be expressed in terms of $\bar{V}^{(k)}$ and $\bar{U}^{(k)}$ in a form analogous to (\ref{solsemialt}).

\subsection{Approximation of the particle system}\label{approximation}
In the same spirit, we may also approximate, for each $n$, the particle system described by (\ref{funris}),  (\ref{sistpart}) by a suitable sequence of systems. For  an initial datum $\mathbf{z}_n$ for the $n$ particles phase space position and configuration of obstacles $\mathbf{c}_M$ ($M= 0$ in the absence of obstacles), we define this system, for integers $M,k\geq 1$, $j=1,\ldots,n$ and $\ell=1,\ldots,M$,  as:

\begin{equation}
\label{approxlevk}
\begin{array}{l}
\xi^{(0)}_{n, \var, 0}(t,z_j)=\xi^{(k)}_{n, \var, 0}(t,z_j)=1
\\
\\
V^{(0)}_{n,\var, M} (t,c_\ell)=0, \quad
\eta^{(0)}_{n, \var, M} (t,c_\ell)=1,\quad \xi^{(0)}_{n, \var, M} (t,z_j)=1\\
\\
V^{(k)}_{n, \var, M} (t, c_\ell)= \cfrac{1}{n} \Som_{i=1}^n \int_0^t ds q_n (x_i (s) -c_\ell)\xi_{n, \var, M}^{(k-1)} (s,z_i) \\
\\
\eta^{(k)}_{n, \var, M} (t,c_\ell)=I_{\{V^{(k)}_{n, \var, M} (t, c_\ell) < \tau_\ell\}}\\
\\
\xi^{(k)}_{n, \var, M} (t,z_j)=I_{\{ x_j (s)\notin \Bigcup_{h=1}^M B_{\var}(c_h)\eta^{(k-1)}_{n,\var, M}(s,c_h)\,\,\,\forall \in [0,t)\}}.
\end{array}
\end{equation}
This sequence is a linearization of the original particle system (which is its formal limit when $k\to\infty$)  in the same  way as (\ref{sistk}) is a linearization of (\ref{sistlim}) and it satisfies what is called in \cite{NOR}  \textit{sandwiching} property (see (\ref{panino1})). This property implies, in a given asymptotics for $n$ and $\var$:

\begin{equation}
\label{impl}
\begin{array}{ll}
\E^n|V^{(k)}_{n,\var,M}-\bar{V}^{(k)}|
\to
0 &\implies\E^n|V_{n,\var,M} -V_L|
\to
0 \\
\\
\E^n|\xi_{n,\var, M}^{(k)}-\bar{\xi}^{(k)}|
\to
0& \implies \E^n\E_{\tau_p}|\xi_{n,\var, M}- \xi_L|
\to
0
\\
\\
\E^n\E_{\tau}|\eta_{n,\var, M}^{(k)}-\bar{\eta}^{(k)}|
\to
0 &\implies \E^n\E_{\tau}|\eta_{n,\var, M}- \eta_L|
\to
0\\
\\
\E^n|\int
dz \mu_n \phi (\xi_{n,\var, M}^{(k)} -\bar{\xi}^{(k)})|
\!\!\to\!
0&\implies\E^n\E_{\tau_p}
|\!\int\!\!
dz \mu_n \phi(\xi_{n,\var,M}- \xi_L)|
\!\to\!
0\\
\end{array}
\end{equation}
and using these implications we shall be able to bypass the direct evaluation of quantities  related to the particle system (\ref{funris}),(\ref{sistpart}).

The advantage in dealing with the  two approximating sequences of systems defined by  (\ref{sistk}) and (\ref{approxlevk})
instead of
the original systems
is that,
for each $k$, the two components evolve in a given field, associated to the functions defined at the previous step $k-1$ in the sequence, and the original self-consistent structure is lost.

The next step would be then to show that (\ref{sistk}) and (\ref{approxlevk}) are, for a given $k$, asymptotically equivalent. Unfortunately, the system defined from
(\ref{approxlevk}), when $\var$ is positive, still keeps strong correlations between all light particles and obstacles positions (in the phase space), correlations which are absent in (\ref{sistk}); this makes hard the direct comparison of the two systems. Therefore, we need to define an intermediate system in which correlations among light particles and obstacles are further reduced.
We shall prove then that this system is equivalent in the limit $n\to\infty$ both to (\ref{sistk}) and to (\ref{approxlevk}).

\subsection{A system asymptotically equivalent to both approximating systems}
For  an initial datum $\mathbf{z}_n$ for the $n$ particles phase space position, and for configurations of obstacles $\mathbf{c}_M$ ($M= 0$ in the absence of obstacles), we define, for integers $M,k\geq 1$, $j=1,\ldots,n$ and $\ell=1,\ldots,M$, the intermediate system in the following way

\begin{equation}
\label{cap}
\begin{array}{l}
\hat{\xi}^{(0)}_{n,\var,0} (t,z_j)=\hat{\xi}^{(k)}_{n,\var,0} (t,z_j)=1
\\
\\
\hat{\xi}^{(0)}_{n,\var,M} (t,z_j)=1 ,
\\
\\
{\hat\xi}^{(k)}_{n,\var,M} (t,z_j)=I_{\{x_j(s) \notin \Bigcup_{h=1}^{M} B_\var (c_h)I_{\{{\bar V}^{(k-1)}(s,c_h)< \tau_{h}\}}  \quad \forall s \in [0,t)\}}\\
{\hat A}^{(k)}_{n,\var,M} (t,c_\ell)= \cfrac 1n \Som_{i=1}^n \int_0^{t} ds\, q_n (x_i (s) -c_\ell) {\hat \xi}^{(k-1)}_{n,\var,M} (s, z_i)\\
\\
 {\hat\eta}^{(k)}_{n,\var,M} (t,c_{\ell})=I_{\{{\hat A}^{(k)}_{n,\var,M}(t,c_{\ell})< \tau_{\ell}\}}
\end{array}
\end{equation}
where $\tau_k$, $1\leq k \leq M$, are independent exponentially distributed times.

In (\ref{cap}), the life function of a light particle at level $k$ is defined through fictitious obstacles life functions ($I_{\{{\bar V}^{(k-1)}(t,c) < \tau_{c}\}}$), corresponding to the obstacles life functions at level $k-1$ associated to (\ref{sistk});
the life function of an obstacle at level $k$, $I_{\{\hat{A}^{(k)}_{n,\var, M}< \tau_c\}}$, is defined through the light particles life functions at level $k-1$. In this way, the correlation between light particles and obstacles is weaker, compared to the same correlation in system (\ref{approxlevk}), and
this allows us to prove, when $n$ and $\var$ verify conditions (\ref{conditio}) and (\ref{conditiovar}),
the convergence
to both systems  (\ref{sistk}) and (\ref{approxlevk})
in quadratic mean w.r.t. the expectation value $\E^n$.

\subsection{Proof of the main theorem}

We proceed now with the proof of our main theorem. To this purpose, we assume we have already proved the sandwiching property for the system (\ref{approxlevk}) and  the asymptotic equivalence of this last system to (\ref{sistk}): we shall postpone to next sections the proof of lemmas and propositions concerning these
two facts, since they are the core of the derivation; as pointed out in section \ref{approximation}, the convergence of (\ref{approxlevk}) to (\ref{sistk}) will be obtained passing through the equivalence in this asymptotics of (\ref{cap}) to both systems  (\ref{approxlevk}) and (\ref{sistk}).

\begin{proof}[Proof of Theorem \ref{teorema}]

Given the definitions in the previous paragraphs, we write
$$
\xi_{n,\var,M}=(\xi_{n,\var,M}-\xi_L) + (\xi_L - \bar{\xi}^{(k)}) +\bar{\xi}^{(k)}
$$
and
$$f=(f-f^{(k)})+f^{(k)}.$$
Then, using formula (\ref{solsemiexpl}) for $f$ and (\ref{semiexplk}) for $f^{(k)}$,
we can write,
for all $\phi\in C_b(\R^d\times\R^d)$ and for all $k\geq 1$:
\begin{eqnarray}
\nonumber
\lefteqn{
\E^n \left[\bigl|\frac 1n \sum_{i=1}^{n} \phi(T^t(z_i))\xi_{n,\var,M}(t,z_i)-\int_{\R^d\times\R^d} dx dv \phi(x,v) f(t, x, v)\bigr|\right]\leq}
\\
\label{fmu1}\\
\nonumber
&&
\E^n \E_{\tau_p}\left[\bigl|\frac 1n
\sum_{i=1}^{n} \phi(T^t(z_i))(\xi_{n,\var,M}
-\xi_L) (t,z_i)
\bigr|\right]+\\
\label{fmu2}\\
\nonumber
&&
\|\phi\|_{\infty}
\left\{\E^n \E_{\tau_p}\left[\frac 1n \sum_{j=1}^{n}|\xi_L
- \bar{\xi}^{(k)}|(t,z_j)
\right] \right.+\\
\nonumber
&&\left.
\int_{\R^d\times\R^d}\!\!\! dx dv C_d|v|f_0(x,v)\int_0^t ds |\si
 -\si^{(k-1)}|(s,x(s))
\right\}+\\
\label{fmu4}\\
\nonumber
&&\E^n \E_{\tau_p}\left[\bigl|\frac 1n \sum_{i=1}^{n}\phi(T^t(z_i))\bar{\xi}^{(k)}(t,z_i)-\!\!\!\!\int_{\R^d\times\R^d}
dx dv f^{(k)}(t,x,v) \phi(x, v)\bigr|\right].
\end{eqnarray}

The term (\ref{fmu1}) vanishes on a set of full measure w.r.t. $\Scr{P}$ as a consequence of the sandwiching property valid for (\ref{approxlevk}) (cf. Corollary \ref{corminor}, Sec. \ref{capitolopanino}) and of the
convergence of (\ref{approxlevk}) to (\ref{sistk}) in the asymptotics specified in (\ref{conditio}),
(\ref{conditiovar}) (cf. Proposition \ref{capbarfin}, Sec. \ref{capitolofinale}).

Thanks to (\ref{conditiomomj}), and in particular to the convergence $|v|\mu^0_n \deb |v|f_0$, the term in curly brackets (\ref{fmu2}) is bounded (excepted on a set of zero measure with respect to $\Scr{P}$) by
\mbox{$K_a \|\si-\si^{(k-1)}\|_{L^{\infty}([0,T]\times\R^d)}$}, so that, by choosing a suitable $k$, it can be made arbitrarily small because of (\ref{convergenze}).

Using Cauchy-Schwarz's inequality, we get then for (\ref{fmu4})
$$
\!\!\!\left|\E^n \E_{\tau_p}\left[\bigl|\frac 1n \sum_{i=1}^{n}\phi(T^t(z_i))\bar{\xi}^{(k)}(t,z_i)-\!\!\!
\int_{\R^d\times\R^d}\!\!\! dx dv f^{(k)}(t,x,v) \phi(x, v)\bigr|\right]\right|^2 \leq
$$
$$
\E^n \left[\bigl|\frac 1n \sum_{i=1}^{n} \phi(T^t(z_i))\E_{\tau_p}[\bar{\xi}^{(k)}(t,z_i)]\bigl|^2 \right]+\int_{\R^d\times\R^d}\!\!\!
dx dv f^{(k)}(t,x,v)\phi(x, v)\times
$$
$$
\times\left[\int_{\R^d\times\R^d}
dx dv f^{(k)}(t,x,v) \phi(x, v)- \frac 2n \sum_{i=1}^{n} \phi(T^t(z_i))\E_{\tau_p}[\bar{\xi}^{(k)}(t,z_i)]\right].
$$
Because of the hypothesis on $f_0$, $\phi(T^t(\cdot))\E_{\tau_p}[\bar{\xi}^{(k)}(t,\cdot)]\in C_b(\R^d\times\R^d)$, so that,
 thanks to (\ref{conditiomomj}) (and in particular to $\mu^0_n\deb f_0$),
on a set of full measure w.r.t. $\Scr{P}$ we have
$$\frac 1n \sum_{i=1}^{n} \phi(T^t(z_i))\E_{\tau_p}[\bar{\xi}^{(k)}(t,z_i)]\to
\int_{\R^d\times\R^d}
dx dv f^{(k)}(t,x,v)\phi(x, v),$$  and (\ref{fmu4}) vanishes on this set.

In the same way, by writing
$$\eta_{n,\var,M}=(\eta_{n,\var,M}-\eta_L)+(\eta_L -\bar{\eta}^{(k)})+\bar{\eta}^{(k)}$$
and
$$\si=(\si-\si^{(k)}) + \si^{(k)},$$ we get,
for all $\psi\in C_K(\R^d)$ and for all $k\geq 1$,
\begin{eqnarray}
\nonumber
\lefteqn{\E^n\left[\bigl|\var^{d-1} \sum_{i=1}^{M}\psi(c_i)\eta_{n,\var,M}(t, c_i)-\int_{\R^d} dx \psi(x)\si(t,x)\bigr|\right]\leq}
\\
\label{fsi1}\\
\nonumber
&&\|\psi\|_{\infty}\left\{\E^n \left[\var^{d-1} \sum_{i=1}^{M}I_{\mathrm{supp}\psi}(c_i)|\eta_{n,\var,M}
-\eta_L|(t,c_i)
\right]\right.+\\
\label{fsi2}\\
\nonumber
&&\left.\!\!\left( \E^n \left[\var^{d-1} \sum_{i=1}^{M}I_{\mathrm{supp}\psi}(c_i)|\bar{\eta}^{(k)}
-\eta_L|(t,c_i)
\right]+
\int_{\mathrm{supp} \psi}dx|\si^{(k)}
-\si|(t,x)\right)
\right\}\\
\label{fsi4}\\
\nonumber
&&+ \E^n \left[\bigl|\var^{d-1} \sum_{i=1}^{M}\psi(c_i)\bar{\eta}^{(k)}(t,c_i)-\int_{\R^d} dx \psi(x) \si^{(k)}(t,x) \bigr|\right],
\end{eqnarray}
where the term in round brackets (\ref{fsi2}) is (everywhere) bounded by 
$K_b\|\int dv (f-f^{(k-1)})\|_{L^{\infty}([0,T]\times\R^d)}$
and
can be made arbitrarily small because of (\ref{convergenze}).

The term (\ref{fsi1}) vanishes on a set of full measure w.r.t. $\Scr{P}$ for the same reasons as (\ref{fmu1}), i.e. as a consequence of the sandwiching property valid for (\ref{approxlevk}) and of the convergence of (\ref{approxlevk}) to (\ref{sistk}) in the asymptotics specified in (\ref{conditio}),
(\ref{conditiovar}).

As for (\ref{fsi4}), we can write
$$
\left|\E^n \left[\bigl|\var^{d-1} \sum_{i=1}^{M}\psi(c_i)\bar{\eta}^{(k)}(t,c_i)-\int_{\R^d} dx \psi(x) \si^{(k)}(t,x)\bigr|\right]\right|^2\leq
$$
$$
\left|\int_{\Lambda_n} dc \psi(c) \si^{(k)}(t,c)-\int_{\R^d} dc \psi(c) \si^{(k)}(t,c)\right|^2 + \var^{d-1}\int_{\Lambda_n} dc \psi^2(c) \si^{(k)}(t,c).
$$
Since $\psi\in C_K(\R^d)$, we have $\psi^h (\cdot)\si^{(k)}(t,\cdot)\in L^1(\R^d)$) for $h=1,2$ and when $n$ grows to infinity we have both $\Lambda_n \to \R^d$ and $\var\to0$,  so that this term vanishes, in the chosen asymptotics.

Collecting all the assertions about the different terms, the theorem is proved.
\end{proof}

\section{Asymptotic equivalence of (\ref{approxlevk}) and (\ref{sistk}) and convergence of the particle system to the limit system}

We collect in the present section the Lemmas and Propositions which will help us to build up the last step of the proof of our main theorem.

\subsection{A few useful bounds}
Let us list a few bounds and formulas which we shall use often in our
calculations.

\bigskip

It will be useful to adopt in next sections, for stochastic variables $\lambda^{\Cal{F}_{n,M}}$ of the form
$$
\lambda^{\Cal{F}_{n,M}}(t,z_j)=I_{\{x_j(s) \notin \bigcup\limits_{h=1}^{M} B_\var (c_h) I_{ \{\Cal{F}_{n,M}
(s,\, c_h)<\tau_h\}}   \quad \forall s \in [0,t)\}},
$$
(with $\Cal{F}_{n,M}(s,c_h)=\Cal{F}_{\mathbf{c}_{M}, \mathbf{\tau}_M, \mathbf{z}_n}(s,c_h)$ a given random function)
the following representation:
\begin{equation}
\label{formaxi}
\lambda^{\Cal{F}_{n,M}}(t,z_j)=
\prod_{h=1}^M \left(1-I_{\{c_h \in {\Cal T}_\var (t,z_j)\}} I_{\{\Cal{F}_{n,M}(T_{z_j, c_h},\,c_h) < \tau_{c_h}\}}\right).
\end{equation}

Given any couple of stochastic variables $\lambda^{\Cal{F}^{(1)}_{n,M}}$, $\lambda^{\Cal{F}^{(2)}_{n,M}}$ of the form (\ref{formaxi}), using
the trivial inequality
$$
|\prod_h I_{A_h} -\prod_h I_{B_h}|\leq
\sum_h |I_{A_h} - I_{B_h}|,
$$
for $M\geq 1$, we get the bound
\begin{equation}
\label{riduzione}
\left|\lambda^{\Cal{F}^{(1)}_{n,M}}_1
-\lambda^{\Cal{F}^{(2)}_{n,M}}_2 \right|(t,z_j)
\leq
\sum_{h=1}^M I_{c_h \in {\Cal T}_\var (t,z_j)}
|\eta^{\Cal{F}^{(1)}_{n,M}}
- \eta^{\Cal{F}^{(2)}_{n,M}}|(T_{z_j,c_h},c_h),
\end{equation}
where
$$
\eta^{\Cal{F}^{(i)}_{n,M}}(T_{z_j,c_h},c_h)= I_{\{\Cal{F}^{(i)}_{n,M}(T_{z_j,c_h},c_h) <\tau_{c_h}\}},\qquad i=1,2 .
$$
The bound (\ref{riduzione}) will be largely used in the last paragraph of this section.

\bigskip

Moreover, because of the hypothesis on $f_0$ (assigning initially uniformly bounded number and kinetic energy limit densities for the light particle component,
cf. Remark \ref{rdebotrasl} in Appendix \ref{app2}),
$\Scr{P}$-a.e. with respect to $Z_\infty$, we may use the bounds
\begin{eqnarray}
\label{somma1}
\|\frac 1n \sum_{j=1}^{n} q_n(x_{j}(s)- c)\|_{\infty}&\leq&K_1\\
\label{somma2}
\|\frac 1n \sum_{j=1}^{n} q_n(x_{j}(s)- c) [{\Cal T}_\var (s, z_{j})]_L\|_{\infty}
&\leq& \var^{d-1} C_d T K_2
\end{eqnarray}
where $K_1, K_2 >0$ are constant independent of $n$.

\subsection{Quantities related to flow tubes}
When bounding correlations in our particle system, we shall need to evaluate expectation values with respect to the Poisson distribution (\ref{poisson}) and to $\mu_n^0$ on volumes corresponding to the intersections of flow tubes.

Defining, for two vectors $v,w\in\R^d$, $\cos \alpha(v,w)=\frac{v\cdot w}{|v||w|}$ (and $\alpha_{ij}=\alpha(v_i,v_j)$), and recalling the
notation $B_r(p)$ introduced at the beginning of Section \ref{sezdue} to denote the ball of radius $r$ centered in $p$, we have, for the intersection of two flow tubes associated to the particles with initial phase space positions $z_i$ and $z_j$, the following trivial bound:
\begin{lemma}
\label{lemmatubo}
For any $\beta\in(0,\frac 1d)$:
\begin{equation}
\label{mistuflu2}
[{\Cal T}_\var (t,z_j)\cap {\Cal T}_\var (t,z_i)]_L\leq
\phantom{C_d \var^{d-1}\min(|v_i|, |v_j|)T I_{|\sin \alpha_{ij}|<\var^{\beta}}}
\end{equation}
$$
[B_2(0)]_L{\var}^{d(1-\beta)} I_{|\sin \alpha_{ij}|\geq\var^{\beta}}
+ C_d \var^{d-1}\min(|v_i|, |v_j|)T I_{|\sin \alpha_{ij}|<\var^{\beta}}.
$$
\end{lemma}

\begin{proof}

We have always
\begin{eqnarray*}
[{\Cal T}_\var (t,z_j)\cap {\Cal T}_\var (t,z_i)]_L&\leq& \min([{\Cal T}_\var (t,z_j)]_L , [{\Cal T}_\var (t,z_i)]_L)\\
\\
&=& C_d \var^{d-1}\min(|v_i|, |v_j|)t
\end{eqnarray*}
and, whenever $\sin \alpha_{ij}\neq 0$, denoting as $y$ the (unique) crossing point for the trajectories $x_i(s)$, $x_j(s)$ (i.e.
$y= x_i+v_i s_1 = x_j +v_j s_2 $ for some $s_1, s_2 \in \R$),
$$
{\Cal T}_\var (t,z_j)\cap {\Cal T}_\var (t,z_i)\subset B_{\frac{2\var}{\sin \alpha_{ij}}}(y).
$$

We can therefore write:

\begin{equation}
|{\mathcal T}_\var (t,z_j)\cap {\mathcal T}_\var (t,z_i)|\leq |B_2(0)|\left(\frac{\var}{\sin \alpha_{ij}}\right)^d \wedge C_d \var^{d-1}\min(|v_i|, |v_j|)t
\end{equation}
and (\ref{mistuflu2}) follows.

\end{proof}

\bigskip

Thanks to the weak convergence of the initial empirical measure toward a regular density, we can estimate the measure (w.r.t. $\mu_n^0$ ) of the set corresponding, for a given velocity $w\in\R^d$, to the grazing collisions. We can prove in fact  the following Lemma:
\begin{lemma}
\label{misI}
Let $\mu_n^0\deb f_0\in \Scr{S}(\R^d\times\R^d)$, $\beta>0$ and $w\in\R^d$. Assume (\ref{conditio}). Denoting
$\cos \alpha_{i}=\frac{v_i\cdot w}{|v_i||w|}$, the following bound
is verified
$$
\frac{1}{n}\sum_{i=1}^{n}
I_{|\sin \alpha_{i}|<\var^{\beta}}\leq K \var^{\zeta}+ o(a_n^{-d})
$$
for any $\zeta\in (0,\frac{d-1}{2} \beta)$.
\end{lemma}

\begin{proof}

We observe first that
we have
$
\Lim_{n\to\infty}\!\frac{\# \{j: |v_j\cdot w|=|w||v_j|\}}{n}\!=\!0,
$
 since \mbox{$f_0\!\in\!\!\Scr{S}(\!\R^d\!\times\!\R^d\!)$},
so that, uniformly in $\var$, the contribution of the corresponding term
vanishes when $n\to\infty$ and it is actually $o(a_n^{-d})$ in the prescribed asymptotics\footnote{In fact, since by Toeplitz's lemma, for all $a\in (0,1]$, $\frac1n\Som_{k=1}^n (\frac{\# \{j: |v_j\cdot w|=|w||v_j|\}}{k})^{a}$ vanishes when $n\to\infty$, we have at least  $\frac{\# \{j: |v_j\cdot w|=|w||v_j|\}}{n}= o(\frac{1}{\sqrt{n}})$ ).}.

We now evaluate the contribution of grazing crossings of particle trajectories.
We choose a  suitable (standard) $C^{\infty}$ regularization   $\bar{R}^{\delta}$ of $I_{|v\cdot w|\neq |v||w|}$, for instance $\bar{R}^{\delta}$ is s.t.:
$$
I_{|v\cdot w|\neq |v||w|}-\bar{R}^{\delta}\leq e^{-\frac{\sin^2 \alpha(v,w)}{\de^2(\de^2 -\sin^2 \alpha(v,w))}}I_{|\sin \alpha(v,w)|\leq \de}.
$$
Here the parameter $\delta$ denotes the radius of the set where the regularization differs from the original characteristic functions (i.e. the set in $\R^d\times\R^d$ s.t.  $|\sin \alpha(v,w)|\leq \de$),
and, for all $\delta>0$ and $0<\iota<\frac{d-1}{2}$, we write:
\begin{eqnarray*}
\lefteqn{\frac{1}{n}\sum_{i=1}^{n} I_{|v_i\cdot w|\neq |v_i||w|}
I_{|\sin \alpha_{i}|<\var^{\beta}}
\leq
\frac{1}{n} \sum_{i=1}^{n}
\frac{\var^{\iota\beta}}{(1-\frac{|v_i\cdot w|^2}{|v_i|^2|w|^2})^{\frac{\iota}{2}}}
{\bar{R}}^{\delta} (v_i) + o(\de)}\\
&&
\leq
\var^{\iota\beta}
\int_{\R^d\times\R^d} dx dv  \frac{f_{0}(x,v)
}{(1-\frac{|v\cdot w|^2}{|v|^2|w|^2})^{\frac{\iota}{2}}} + K_R \frac{\var^{\iota\beta}}{\de^{3+\iota}} + o(\de).
\end{eqnarray*}
The error coming from the regularization of the characteristic function $I_{|v\cdot w|\neq |v||w|}$ is $o(\de)$ because of the weak convergence of $\mu_n^0$ toward $f_0\in \Scr{S}(\R^d\times\R^d)$ and, for the same reason, all constants (here and in the following estimates) are uniform in $w$ and $n$ and depend only on the dimension $d$ and on few $L^p$ norms of $f_0$.
We obtain then, for suitable choices of
$\de$,
$$
\frac{1}{n}\sum_{i=1}^{n} I_{|v_i\cdot w|\neq |v_i||w|}
I_{|\sin \alpha_{i}|<\var^{\beta}}
\leq
K \var^{\zeta},
$$
for any $\zeta\in(0, \iota\beta)$ and therefore for any $\zeta\in(0,\frac{d-1}{2}\beta)$. Collecting the two estimates, the lemma is proved.
\end{proof}

We shall use Lemma \ref{misI} for free choices of $\beta\in(0, \frac 1d)$, so that $\zeta\in (0,\frac{d-1}{2} \beta)\subset (0,\frac12 -\frac1{2d})$.

\begin{rem}
The estimate obtained in Lemma \ref{misI} is valid for any given sequence of empirical measures $\{\mu_n^0\}$ such that $\mu_n^0\deb f_0\in \Scr{S}(\R^d\times\R^d)$ and it is actually a stronger assertion with respect to what we need to prove Theorem \ref{teorema} (which is valid $\Scr{P}-a.e.$ w.r.t. $Z_\infty$). In order to prove Lemma \ref{misI}, where, as stated in the introductory sentence to the lemma, we bound the measure of a set with respect to the empirical measure $\mu_n^0$,
we need to evaluate the measure of the set $\{(x,v): |v\cdot w|=|w||v|\}$ with respect to the measure $\mu_n^0$ because the measure of this set vanishes only asymptotically (i.e. with respect to the limit measure with density $f_0$). We choose to evaluate this measure with respect to the parameter $a_n^d$ (i.e. as $o(a_n^{-d})$) for further convenience.

A similar, simpler statement can be proved $\Scr{P}$-a.e. w.r.t. $Z_\infty$: in this case we obtain the bound $\frac{1}{n}\sum_{i=1}^{n}
I_{|\sin \alpha_{i}|<\var^{\beta}}\leq K \var^{\zeta}$ ($\Scr{P}$-a.e. w.r.t. $Z_\infty$).
This alternative statement could be used instead of Lemma \ref{misI} to prove our main theorem.
We prefer nevertheless to use Lemma \ref{misI}, getting in this way bounds which are (as much as possible) valid on the whole subset of initial sequences $Z_{\infty}$ such that the associated sequence of empirical measures $\{\mu_n^0\}_{n=1}^{\infty}$ converges weakly to $f_0\in \Scr{S}(\R^d\times\R^d)$.
\end{rem}

\begin{rem}
Notice that, since $\frac{1}{n^2}\Som_{i,j: |v_i\cdot v_j|=|v_i||v_j|}1=o (a_n^{-d})$ and all bounds are uniform in the velocity $w$, we get also
\begin{equation}
\label{I}
\frac{1}{n^2}\sum_{i,j=1}^{n}
I_{|\sin \alpha_{ij}|\leq
\var^{\beta}}\leq
\frac1n\sum_{j=1}^{n}
(\frac1n\sum_{i=1}^{n} I_{|\sin \alpha_{ij}|\leq
\var^{\beta}})\leq K \var^{\zeta}+ o(a_n^{-d}).
\end{equation}
\end{rem}

\bigskip

A last very useful tool will be the following parametrization of the points \small{\mbox{$c\!\in\!\Cal T_\var \!(t,\!z)$}}. Let $u=T_{z,c}$ be the maximal collision time defined in
(\ref{tcollis}); then we may define the change of variables
$$
c=x (u) + \var \om, \quad u \in [0,t),\quad \om \in \pa B_1 (0).
$$
Using this parametrization, for each  non negative function $V\in L^{\infty}([0,T]; W^{1,\infty}(\R^d))$ we can write $V(T_{z,c},c)=V\left(u,x (u) + \var \omega \right)$, getting then

\begin{equation}
\label{camvar}
\int_{\Cal T_\varepsilon (t,z)}\!\!\! dc e^{-V(T_{z,c},c)}=
\varepsilon^{d-1} C_d |v|\left(\int_0^t du [e^{-V(u,x (u))} + \phi_{\var}(u,x(u))]\right)
\end{equation}
\noindent
where we may bound the remainder through the trivial inequality \mbox{$|e^{-y}-e^{-x}|\leq |x-y|$}, for $x,y\geq 0$, and the function
$\phi_{\var}$ is such that $\sup_{[0,T]}\|\phi_{\var}\|_{L^{\infty}(\R^d)}=O(\var)\to 0$ when $\var\to 0$.

\subsection{The sandwiching property } \label{capitolopanino}
We prove now the sandwiching property for the system defined by (\ref{approxlevk}) and its implication on the convergence of (\ref{sistpart}).

\begin{lemma} (sandwiching property)
\label{sandwlem}
Consider $ V_{n, \var, M}$,  $\xi_{n, \var, M}$, $\eta_{n, \var, M}$ defined by (\ref{funris}), (\ref{sistpart}) and $\xi_{n, \var, M}^{(k)}$, $V_{n, \var, M}^{(k)}$, $\eta^{(k)}_{n, \var, M}$ defined by  (\ref{approxlevk}). Then, for $k=1,2\ldots$,
\begin{equation}
\label{panino1}
\begin{split}
\xi_{n,\var, M}^{(2k-1)} \leq \xi_{n, \var, M}^{(2k+1)} \leq \xi_{n,\var,M} \leq \xi_{n,\var, M}^{(2k)} \leq \xi_{n, \var, M}^{(2k-2)}\\
\\
V_{n, \var, M}^{(2k-2)} \leq V_{n,\var, M}^{(2k)} \leq V_{n,\var,M} \leq V_{n, \var, M}^{(2k+1)} \leq V_{n, \var, M}^{(2k-1)}\\
\\
\eta^{(2k-1)}_{n, \var, M}\leq \eta^{(2k+1)}_{n,\var, M} \leq \eta_{n, \var, M}\leq \eta^{(2k)}_{n, \var, M} \leq \eta^{(2k-2)}_{n, \var, M}.
\end{split}
\end{equation}
\end{lemma}

\begin{proof}
From definitions (\ref{funris}), (\ref{sistpart}) and (\ref{approxlevk}), we have  $0=V^{(0)}_{n, \var, M} \leq V_{n, \var, M} \leq V^{(1)}_{n,\var, M}$, which implies  $\eta^{(0)}_{n,\var, M}\geq \eta_{n,\var, M}\geq \eta^{(1)}_{n,\var, M}$, which in turn implies $ \xi^{(1)}_{n,\var, M}\leq \xi_{n,\var, M}\leq \xi^{(2)}_{n,\var, M}\leq \xi^{(0)}_{n,\var, M}=1$.

We have therefore, using again these definitions:
$$
\xi_{n,\var, M}^{(1)} \leq \xi_{n, \var, M}^{(3)} \leq \xi_{n,\var,M} \leq \xi_{n,\var, M}^{(2)} \leq \xi_{n, \var, M}^{(0)}
$$
$$
V_{n, \var, M}^{(0)} \leq V_{n,\var, M}^{(2)} \leq V_{n,\var,M} \leq V_{n, \var, M}^{(3)} \leq V_{n, \var, M}^{(1)}
$$
$$
\eta^{(1)}_{n, \var, M}\leq \eta^{(3)}_{n,\var, M} \leq \eta_{n, \var, M}\leq \eta^{(2)}_{n, \var, M} \leq \eta^{(0)}_{n, \var, M}
$$
and (\ref{panino1}) is valid for $k=1$.

Since the following chain is also valid:
\begin{eqnarray*}
V^{(2k-2)}_{n,\var, M}&\leq& V^{(2k)}_{n,\var, M} \leq V_{n, \var, M}
 \leq V^{(2k+1)}_{n, \var, M} \leq V^{(2k-1)}_{n, \var, M} \implies\\
\\
\eta^{(2k-2)}_{n, \var, M}&\geq& \eta^{(2k)}_{n,\var, M} \geq \eta_{n, \var, M}\geq \eta^{(2k+1)}_{n, \var, M} \geq \eta^{(2k-1)}_{n, \var, M} \implies\\
\\
\xi^{(2k-1)}_{n,\var, M}&\leq& \xi^{(2k+1)}_{n, \var, M} \leq  \xi_{n, \var, M} \leq \xi^{(2k+2)}_{n,\var, M} \leq \xi^{(2k)}_{n, \var, M}\implies\\
\\
V^{(2k)}_{n, \var, M}&\leq& V^{(2k+2)}_{n, \var, M} \leq  V_{n, \var, M} \leq V^{(2k+3)}_{n, \var, M} \leq V^{(2k+1)}_{n, \var, M}
\end{eqnarray*}
the statement of the proposition follows by induction.
\end{proof}

We may then prove the following corollary to Lemma \ref{sandwlem}:
\begin{corol}
\label{corminor}
Consider $ V_{n,\var,M}$,  $\xi_{n,\var,M}$, $\eta_{n, \var, M}$ and $V_{n, \var, M}^{(k)}$, $\xi_{n,\var, M}^{(k)}$, $\eta^{(k)}_{n, \var, M}$  defined as in Lemma \ref{sandwlem} and $V_L$ and $\xi_L$ given by (\ref{risslo}) and (\ref{csisl}), with $(f,\si)$ unique solution of (\ref{sistlim}).
Then, when $n\to \infty, \var\to 0$, for any sub-linear operators $L$ on
$L^{\infty}([0,T]\times\R^d\times\R^d)$

$$L|V^{(k)}_{n,\var,M}-\bar{V}^{(k)}|\to 0\implies
L|V_{n,\var,M} -V_L|\to 0
$$

$$
L|\xi_{n,\var, M}^{(k)} -\bar{\xi}^{(k)}|\to 0 \implies L
\E_{\tau_p}|\xi_{n,\var,M}- \xi_L|\to 0
$$

$$
L\E_\tau|\eta_{n,\var, M}^{(k)} -\bar{\eta}^{(k)}|\to 0 \implies L
\E_\tau|\eta_{n,\var,M}- \eta_L|\to 0
$$

$$
L|\int_{\R^d\times\R^d}\!\!\!\!\!\!\!\! dz \mu^{0}_n \phi (\xi_{n,\var, M}^{(k)} -\bar{\xi}^{(k)})|\to 0 \implies L\E_{\tau_p}
|\int_{\R^d\times\R^d}\!\!\!\!\!\!\!\! dz \mu^{0}_n \phi(\xi_{n,\var,M}- \xi_L)|\to 0 .
$$
\end{corol}

\begin{proof}

We denote here by $\phi_+=\phi\wedge 0$,  $\phi_-=-\phi\vee 0$ resp. the positive and the negative part of $\phi$.

Let $\Cal{F}_{n,\var,M}$ be one among the non negative functions $V_{n,\var,M}$, \mbox{$1-\xi_{n, \var, M}$}, \mbox{$1-\eta_{n,\var, M}$}, $\int dz \mu^{0}_n \phi_+\, \xi_{n, \var, M}$, $\int dz \mu^{0}_n \phi_-\, \xi_{n, \var, M}$ and $\Cal{F}^{(k)}_{n,\var,M}$ its approximant of order $k$ (defined through (\ref{approxlevk})). Let $\Cal{F}_L$ be its associate limit function (defined through (\ref{risslo}) and (\ref{csisl})) and $\Cal{F}^{(k)}$ the approximant of $\Cal{F}_L$ (defined through (\ref{risoko}),(\ref{risokp}) and (\ref{vitak})).

Because of (\ref{panino1}), we have, for a given $k\geq 1$:

$$
\Cal{F}_{n,\var,M}^{(2k)}-\Cal{F}_L\leq\Cal{F}_{n,\var,M}  -\Cal{F}_L\leq \Cal{F}_{n,\var,M}^{(2k+1)}-\Cal{F}_L ,
$$

and therefore,

$$
|\Cal{F}_{n,\var,M}  -\Cal{F}_L|\leq |\Cal{F}_L- \Cal{F}^{(2k)}|+
|\Cal{F}_L- \Cal{F}^{(2k+1)}|+
$$

$$
\phantom{|\Cal{F}_{n,\var,M}  -\Cal{F}_L|+}|\Cal{F}^{(2k+1)}-\Cal{F}^{(2k+1)}_{n,\var,M}|+
|\Cal{F}^{(2k)}-\Cal{F}^{(2k)}_{n,\var,M}|.
$$

Since $\|V_L-\bar{V}^{(k)}\|_{L^{\infty}([0,T]\times\R^d)}$ and
$\|\E_{\tau_p}[\xi_L-\bar{\xi}^{(k)}]\|_{L^{\infty}([0,T]\times\R^d\times\R^d)}$ vanish in the $k\to\infty$ limit (cf. (\ref{convergenze})), and  since $\E_{\tau_p}[|\eta_L-\bar{\eta}^{(k)}|]\leq |V_L-\bar{V}^{(k)}|$ (cf. Remark \ref{oss}), the assertion is proved.
\end{proof}

As already pointed out, Corollary \ref{corminor} is valid in particular when $L$ is an expectation value operator.

\bigskip

\subsection{Equivalence between (\ref{cap}) and (\ref{sistk})}
\subsubsection{Motion of light particles in a decaying medium: equivalence of the light particle components of (\ref{cap}) and (\ref{sistk})}
We prove here a lemma which describes the behavior of the light particle component of the system (\ref{cap})  in the  $n\to\infty$ limit. In this lemma, we consider a particle system such that  the rate of death of the obstacles is given and independent from the light particles component, while the light particles are instantaneously absorbed each time they collide with an obstacle. The correspondent limit system consists of equations which are coupled only in a very weak way and the parameters determining the asymptotics are independent.

Let $g(t,c)\in L^1([0,T]; W^{1,\infty}(\R^d))$ be a non negative function. The risk function $V(t,c)$ defined, for $t\leq T$, by
$$
V(t,c)=\int_0^t g (s, c)ds
$$
is non decreasing as a function of $t$.

As in the previous sections, we define:
\begin{equation}
\label{campomedio}
\begin{array}{l}
\eta^V (t, c_h)= I_{\{V(t,c_h) < \tau_h\}}\\
\\
\xi^V_{n,\var,M} (t, z_i )=I_{\{ x_i (s)\notin \Bigcup_{k=1}^M B_{\var}(c_k)\eta^V (s,c_k)\,\,\,\forall \in [0,t)\}} .
\end{array}
\end{equation}

We have then:

\begin{lemma}
\label{lemmauno1}
Let $g(t,c)\in L^1([0,T]; W^{1,\infty}(\R^d))$ and $\xi^V_{n,\var, M}$ be defined by (\ref{campomedio}). Consider the sequence $Z_{\infty}\in (\R^d\times\R^d)^{\infty}$ and its associated sequence of empirical measures $\{\mu_n^0\}_{n=1}^{\infty}$. Assume $\mu_n^0 \deb f_0\in \Scr{S}$ and $|v|\mu^0_n \deb |v| f_0\in \Scr{S}$.
Then $\forall \phi \in C_b(\R^d\times\R^d)$,
$$
\Lim_{\lmul{n \rightarrow \infty\\\var\to 0}} \E^n \left[\left|\frac 1n \sum_{i=1}^n \xi^V_{n,\var,M} (t,z_i) \phi(T^t(z_i)) -\!\!\int_{\R^d\times\R^d}\!\!\!\!\! dx dv f(t,x, v) \phi(x,v)\right|^2\right]\!\!=0
$$
(independently of the order of the limits),
where
$(f,\sigma)$ is the (unique) solution of the following system of partial differential equations:
$$
\left\{
\begin{array}{l}
\pa_t f +v\cdot\nabla_x f = -|v| C_d  \si f \\
\pa_t \si= - g \si\\
f(0,x,v)=f_0 (x,v)\\
\si(0,x)=\mu.
\end{array}
\right.
$$
\end{lemma}

\begin{rem}
\label{ossv}
Of course, when $V=\bar{V}^{(k-1)}$, we have $\xi^V_{n,\var,M}=\hat{\xi}^{(k)}_{n,\var,M}$, so that obviously (by the triangular inequality) Lemma \ref{lemmauno1} implies, when $\var\to0$, $n\to\infty$, the limit
$\E^n[|\int dx dv \mu^0_n (\hat{\xi}^{(k)}_{n,\var,M}- \bar{\xi}^{(k)})\phi|]\to 0$
.
\end{rem}

\begin{proof}

According to the definition of maximal collision time (\ref{tcollis}), we have
$$
\xi^V_{n,\var,M} (t,z_j)=1 \iff
\begin{array}{c}
\forall h:1\leq h\leq M, \,\,\\ \\T_{z_j, c_h} \wedge t \geq \inf_{s \in \R_{+}} \{\eta^V (s,c_h) =0 \}\wedge t.
\end{array}
$$
\smallskip

The function $F(\cdot, c)=\inf\{z: z=V^{-1} (\cdot, c)\}$ is well defined and
$$\inf \{s \in \R_{+}; \eta^V (s,c_h) =0 \}=F(\tau_{c_h}, c_h).$$
Moreover, the maximal collision time $T_{z_j,c}$ is greater than $t$ if and only if $c$ does not belong to the tube
${\Cal T}_\var (t,z_j)$.
It follows that

\begin{eqnarray*}
P^n \left[\xi^V_{n,\var,M} (t,z_j)=1\right]&=&
P^n \left[\{1\leq h\leq M:\,T_{z_j, c_h} \wedge t \geq F( \tau_{c_h},c_h) \wedge t \}\right]\\
\\
&=&\!\!
\E^n \left[\prod_{h=1}^{M}\left(I_{\{c_h \notin {\mathcal T}_\var (t,z_j)\}} + I_{\{c_h \in {\mathcal T}_\var (t,z_j);\,\, T_{z_j, c_h} \geq V^{-1}(\tau_{c_h} , c_h) \wedge t\}}\right)\right]\\
&=&\exp\left[-\mu_\var \int_{{\Cal T}_\var (t,z_j)} dc e^{-V(T_{z_j,c},c)}\right].
\end{eqnarray*}

Using (\ref{camvar}), because of the weak convergence of $\mu_n^0$ and $|v|\mu_n^0$, we get
\begin{equation}
\label{convinmedia}
\E^n\left[\frac 1n \sum_{i=1}^n \xi^V_{n,\var,M} (t,z_i) \phi(T^t(z_i))\right] \stackrel{n\to\infty, \var\to 0}{\rightarrow}  \int_{\R^d\times\R^d} dx dv f(t,x, v) \phi(x,v),
\end{equation}
the result being independent of the ordering of the limits.

From (\ref{convinmedia}), we prove Lemma \ref{lemmauno1} as follows.
We call
\begin{eqnarray*}
\Cal{E}_{1\,n}&=&\E^n \left[\frac 1{n^2}\sum_{\lmul{i,j=1\\i\neq j}}^n \xi^V_{n,\var,M} (t,z_i)  \xi^V_{n,\var,M} (t,z_j) \phi(T^t(z_i))\phi(T^t(z_j))\right]\\
\Cal{E}_{2\,n}&=&\int_{\R^d\times\R^d}\!\!\!\! dx dv f(t,x, v) \phi(x,v)
\!-\E^n\!\! \left[\frac 2n \sum_{i=1}^n \xi^V_{n,\var,M} (t,z_i) \phi(T^t(z_i))\right].
\end{eqnarray*}

We write then:
$$
\E^n \left[\left|\frac 1n \sum_{i=1}^n \xi^V_{n,\var,M} (t,z_i) \phi(T^t(z_i)) - \int_{\R^d\times\R^d}\!\!\!\!\!\! dx dv f(t,x, v) \phi(x,v)\right|^2\right]=
$$
$$
\E^n \left[\frac 1{n^2}
\sum_{i=1}^n \xi^V_{n,\var,M} (t,z_i) \phi^2(T^t(z_i))\right]
+ \Cal{E}_{1\,n} +\Cal{E}_{2\,n}
\int_{\R^d\times\R^d}\!\! dx dv f(t,\!x,\!v) \phi(x,\!v)
$$
where
$\E^n \left[\frac 1{n^2}\Som_{i=1}^n \xi^V_{n,\var,M} (t,z_i) \phi^2(T^t(z_i))\right]$,  when \mbox{$n\to\infty$}, obviously vanishes and, thanks to (\ref{convinmedia}),
$$
\Cal{E}_{2\,n} \int_{\R^d\times\R^d}\!\! dx dv f(t,x, v) \phi(x,v)\to -
\left(\int_{\R^d\times\R^d}\!\! dx dv f(t,x, v) \phi(x,v)\right)^2 .
$$

As for
$\Cal{E}_{1\,n}$
 we observe that, for $i\neq j$:
$$
\E^n\left[\xi^V_{n,\var,M} (t,z_i)  \xi^V_{n,\var,M} (t,z_j)\right]=
\phantom{-\mu_{\var}\int_{{\Cal T}_\var (t,z_i)} dc\,
e^{-V(T_{z_i,c},c)}-\mu_{\var}\int_{{\Cal T}_\var (t,z_j)} dc\,
}
$$
$$
\exp\left[-\mu_{\var}\int_{{\Cal T}_\var (t,z_i)}\!\!dc\,
e^{-V(T_{z_i,c},c)}-\mu_{\var}\int_{{\Cal T}_\var (t,z_j)}\!\! dc\,
e^{-V(T_{z_j,c},c)}\right]
\times
$$
$$
\exp\!\!\left[\mu_{\var}\int_{{\Cal T}_\var (t,z_j)\bigcap {\Cal T}_\var (t,z_i)} \!\!\!\!\!\!\!\!\!dc
\left(I_{\{c:T_{z_i,c}\leq T_{z_j,c} \}}
e^{-V(T_{z_j,c},c)}
+I_{\{c:T_{z_i,c}> T_{z_j,c} \}}
e^{-V(T_{z_i,c},c)}\right)
\right]\!\!.
$$

We write
$
1=
I_{|\sin \alpha_{ij}|\geq\var^{\beta}}
+
I_{|\sin \alpha_{ij}|<\var^{\beta}},
$
 for $\beta <\frac 1d$.
Then, using (\ref{I}) and $[{\Cal T}_\var (t,z_j)\cap {\Cal T}_\var (t,z_i))]_L\leq \bar{K}_{\mu}{\var}^{d-1}$ (always valid for intersections of flow tubes, as shown also in the proof of Lemma \ref{lemmatubo})
we have:
\begin{eqnarray*}
\lefteqn{\frac{1}{n^2}\sum_{\lmul{i,j=1\\i\neq j}}^{n}
I_{|\sin \alpha_{ij}|<\var^{\beta}}
\E^n \left[\xi^V_{n,\var,M} (t,z_i)  \xi^V_{n,\var,M} (t,z_j)\right]}\\
&&\leq \frac{1}{n^2} \sum_{\lmul{i,j=1\\i\neq j}}^{n}
e^{\mu_{\var}[{\Cal T}_\var (t,z_j)\bigcap {\Cal T}_\var (t,z_i)]_L}
I_{|\sin \alpha_{ij}|<\var^{\beta}}
\leq K_{\mu} \var^{\zeta} + o(a_n^{-d}).
\end{eqnarray*}

As for the remaining part (since $I_{|\sin \alpha_{ij}|\geq\var^{\beta}}=1-
I_{|\sin \alpha_{ij}|<\var^{\beta}}$), from (\ref{I}) and Lemma \ref{lemmatubo} we get:
$$
\frac{1}{n^2}\sum_{\lmul{i,j=1\\i\neq j}}^{n} \phi(T^t(z_i))\phi(T^t(z_j))
I_{|\sin \alpha_{ij}|>\var^{\beta}} \E^n \left[\xi^V_{n,\var,M} (t,z_i)  \xi^V_{n,\var,M} (t,z_j)\right]=
$$
$$
o(a_n^{-d})+O(\var^{\zeta})+
\frac{1}{n^2}\!\sum_{i,j=1}^{n} \phi(T^t(z_i))\phi(T^t(z_j))
\times
$$
$$
\times
\exp\! \left[-\mu_{\var}\!\left(\int_{{\Cal T}_\var (t,z_i)} \!dc
e^{-V(T_{z_i,c},c)}
+\!\!\int_{{\Cal T}_\var (t,z_j)} \!dc
e^{-V(T_{z_j,c},c)}\right)\!\!+\!\! K_t \var^{1-d\beta}\right]
$$
since in this case $\mu_{\var} [{\Cal T}_\var (t,z_j)\bigcap {\Cal T}_\var (t,z_i))]_L=O(\var^{1-d\beta})\to 0$, so that (using  (\ref{camvar}) once more)
$$
\Cal{E}_{1\,n}
\to
\left( \int_{\R^d\times\R^d} dx dv f(s,x, v) \phi(x,v)\right)^2
$$
and Lemma \ref{lemmauno1} is proved.
\end{proof}

We emphasize again that all results in this Lemma are independent of the ordering of the limits $n\to\infty$, $\var\to 0$.

\subsubsection{Equivalence of the obstacle components of (\ref{cap}) and (\ref{sistk})}
Lemma \ref{lemmauno1} describes the asymptotic average behavior of the light particle component of the system (\ref{cap}). The asymptotic average behavior of its risk function, for an obstacle located in $c$, should be equivalent to the behavior associated  to the risk (\ref{risoko}). This fact is described by the proposition proved in this section and its corollary (notice that the risk function of a given obstacle with position $c_h$, differently from the generic risk function associated to a given position, is defined only if there exists an obstacle in $c_h$ ). The second proposition will be useful in the proof of equivalence of (\ref{cap}) to (\ref{approxlevk}), in the next section.

\smallskip

It will be convenient to define the following stochastic functions.
For $z=(x,v)\in\R^d\times\R^d$:
\begin{equation}
\label{rhodef}
\rho^{(k)}_{\var}(t,c,z)= 1-I_{\{c \in {\Cal T}_{\var} (t,z)\}}\, I_{\{{\bar V}^{(k)}(T_{z,c},c) < \tau_c\}}
\end{equation}

$$
\rho^{(k)}(t,z)=\int_{0}^{t}du \int_{\partial B_{1} (0)} |v\cdot\omega|d\omega \exp\left(-{\bar V}^{(k)} (u,x (u))\right)
$$

$$
Q_n^{ij}(s,t,c)=q_n (x_i (s) -c) q_n (x_j (t)-c)
$$

$$
\Cal{R}_{\var}^{ij\,(k)}(s,t,c)=\E_{\tau_c}[\rho^{(k)}_{\var} (s,c,z_i)\rho^{(k)}_{\var} (t,c,z_j)].
$$

\bigskip

Given these definitions, a fundamental tool in all estimates in this section is the following lemma:

\begin{lemma}
\label{lem:ind2}

Take a sequence $Z_{\infty}\in (\R^d\times\R^d)^{\infty}$ and its associated sequence of empirical measures $\{\mu_n^0\}_{n=1}^{\infty}$. Take a sequence $\{\Lambda_n\}_{n=1}^{\infty}$ of bounded, Lebesgue measurable sets such that
$B_n\subset \Lambda_n$ (with $B_n$ defined in (\ref{bn})) and
assume $\mu_n^0 \deb f_0\in \Scr{S}$ and $|v|\mu_0^n \deb |v| f_0\in \Scr{S}$. For $0 \leq s \leq t$, we have
\begin{equation}
\label{intrho}
\begin{array}{l}
\cfrac{1}{[\Lambda_n]_L}\int_{\Lambda_n} dc\,
\Cal{R}_{\var}^{ij\,(k)}(t,s,c)
=1-\cfrac{\var^{d-1}}{[\Lambda_n]_L} [\rho^{(k)}(t,z_i)+\rho^{(k)} (s,z_j)(1-\de_{ij})]\\
\phantom{\cfrac{1}{[\Lambda_n]_L}\int_{\Lambda_n} dc\,
\Cal{R}_{\var}^{ij\,(k)}(t,s,c)=}+\cfrac{R^{(2)}_{\var} (z_i,z_j,t,s)}{[\Lambda_n]_L}
\end{array}
\end{equation}
where, for each $\zeta\in(0,\frac12-\frac1{2d})$,
\begin{equation}
\label{ro0}
\cfrac{1}{n} \sum_{i=1}^{n} \|R^{(2)}_\var(z_i,z_j)\|_{L^{\infty}([0,T]\times[0,T])}\leq K_{r1}\var^{d-1+\zeta} |v_j|
+\var^{d-1}[o(a_n^{-d})]
\end{equation}

\noindent
and

\begin{equation}
\label{ro2}
\Lim_{\var \to 0} \var^{1-d} \cfrac{1}{n^2} \sum_{i,j=1}^{n} \|R^{(2)}_\var(z_i,z_j)\|_{L^{\infty}([0,T]\times[0,T])}=0.
\end{equation}

\end{lemma}

\begin{proof}

Formula (\ref{intrho}) follows, through simple calculations,
from (\ref{camvar}), expanding the expectation value in the definition of $\Cal{R}_{\var}^{ij\,(k)}$, given in (\ref{rhodef}), and writing
$$\int_{\Cal T_\varepsilon (t,z)} dc e^{-\bar{V}^{(k)}(T_{z,c},c)}=
\varepsilon^{d-1}\left[\rho^{(k)}(t,z)+ C_d |v|\int_0^t du \phi_{\var}(u,x(u))\right];$$ doing this, we obtain for the remainder
\begin{equation}
\label{erre2}
\begin{split}
|R^{(2)}_{\var} (z_i,z_j,t,s)|
\leq\var^{d-1} C_d (|v_i|+|v_j|)\|\varphi_{\var}\|_{L^{\infty}([0,T]\times\R^d)} \\
\\
+[{\Cal T}_{\var}(s,z_j) \cap {\Cal T}_{\var} (t,z_i)]_L
\end{split}
\end{equation}
where $\|\varphi_{\var}\|_{L^{\infty}([0,T]\times\R^d)}=O(\var)\to 0$.

We use then the bound (\ref{erre2}) to obtain (\ref{ro0}) and (\ref{ro2}). In the single sum in (\ref{ro0}), so as in the double sum in (\ref{ro2}), thanks to the weak convergence of $|v|\mu^0_n$, the contribution of the first term on the right-hand side of (\ref{erre2}) is $O(\var^{d})$. As for the contribution of $[{\Cal T}_{\var}(s,z_j) \cap {\Cal T}_{\var} (t,z_i)]_L$, we bound the Lebesgue measure of the intersection of the tubes using Lemma \ref{lemmatubo};  then, thanks to the weak convergence of $\mu^0_n$ toward the (sufficiently) regular function $f_0$, we can use Lemma \ref{misI} and (\ref{I}), with $\beta<\frac 1d$, to estimate the contribution coming from the grazing collisions (being the remaining part in both sums $o(\var^{d-1})$).  So, from Lemma \ref{misI}  we get straightforwardly (\ref{ro0}) and, from (\ref{I}), we get the bound:
\begin{equation}
\label{R2}
\cfrac{1}{n^2} \sum_{i,j=1}^{n} \|R^{(2)}_{\var} (z_i,z_j)\|_{L^{\infty}([0,T]\times[0,T])}
\leq K_{R} \var^{d-1}[\var^{\zeta}+o(a_n^{-d})],
\end{equation}
and, when $\var\to 0$, we obtain (\ref{ro2}).

\end{proof}

We may now prove the main proposition of this section:

\begin{prop}
\label{lem:ind}
Consider the stochastic variables defined in (\ref{cap}) and (\ref{risoko}) and the sequences of positive real numbers $\{a_n\}_{n=1}^{\infty}$ and $\{\var_n\}_{n=1}^{\infty}$ satisfying conditions (\ref{conditio}) and (\ref{conditiovar}), i.e. $\Lim_{n\to\infty}a^ d_n \, n^{-\kappa} =0$, for some $\kappa\in(0,\frac 12)$, and $\Lim_{n\to\infty} a_n^d \var^{\zeta} = 0$, for some $\zeta\in (0,\frac12-\frac1{2d})$.

Then, $\Scr{P}$-almost everywhere w.r.t. sequences of initial data $Z_\infty$,
$$
\Lim_{n\rightarrow \infty}[\Lambda_n]_L\sup_{t\in[0,T]} {\E^n}\left[\left|{\hat A}^{(k)}_{n,\var,M} (t,c)-\bar{V}^{(k)}(t,c)\right|^2 I_{M\geq 1}\right]=0
$$
where $\{\Lambda_n\}$ is an increasing sequence of bounded Lebesgue measurable sets such that $B_n \subset \Lambda_n $ (with $B_n$ defined in (\ref{bn})) and $\Lim_{n\to\infty}\Lambda_n =\R^n$.
\end{prop}

\begin{proof}

We expand the square power in the expectation value and we calculate the value of the three resulting terms.

Because of definition (\ref{rhodef}), using (\ref{formaxi}), we can write:
$$
\hat{\xi}^{(k-1)}_{n,\var,M}(s,z_j)=
\prod_{h=1}^M  \rho^{(k-2)}_{\var} (s,c_h,z_j).
$$
We substitute this expression in the definition of ${\hat A}^{(k)}_{n,\var,M}$, in system (\ref{cap}), and we obtain (resumming on the Poissonian variables):
$$
[\Lambda_n]_L{\E^n}\left[\left|{\hat A}^{(k)}_{n,\var,M} (t,c)\right|^2\!\!\! I_{\{M\geq 1\}}\right]
=
$$
$$
\frac1{n^2}\sum_{i,j =1}^n \int_{[0,t]^2}\!\!\!\! ds_1\, ds_2 \left(\exp\left[\mu_{\var} [\Lambda_n]_L\big(\int_{\Lambda_n}\!dc\,
\cfrac{\Cal{R}_{\var}^{ij\,(k-2)}(s_1,s_2,c)}{[\Lambda_n]_L}-1\big)\right]
\!\!-
e^{-\mu_{\var} [\Lambda_n]_L}\right)
\times
$$
$$
\times\;\frac{\int_{\Lambda_n} dc\, Q_n^{ij}(s_1,s_2,c)\Cal{R}_{\var}^{ij\,(k-2)}(s_1,s_2,c)}{\int_{\Lambda_n}\!dc\,
\cfrac{\Cal{R}_{\var}^{ij\,(k-2)}(s_1,s_2,c)}{[\Lambda_n]_L}}.
$$

Then, using Lemma \ref{lem:ind2}, we have:
\begin{equation}
\label{alfan}
[\Lambda_n]_L{\E^n}\left[\left|{\hat A}^{(k)}_{n,\var,M} (t,c)\right|^2\!\!\! I_{\{M\geq 1\}}\right]
=
\end{equation}

$$
\frac1{n^2}\sum_{i,j =1}^n \int_{[0,t]^2}\!\!\!\! ds_1\, ds_2\,
e^{-\mu [\rho^{(k-2)}(s_1,z_i)+\rho^{(k-2)}(s_2,z_j)(1-\de_{ij})]}
$$

$$
\times(1+\var^{1-d}R^{ij(3)}_{\var})
\left[\int_{\Lambda_n} dc \, Q_n^{ij}(s_1,s_2,c)\Cal{R}_{\var}^{ij\,(k-2)}(s_1,s_2,c)\right]
$$
where $R^{ij(3)}_{\var}=R^{(3)}_{\var}(s_1,s_2,z_i,z_j)\leq K_t R^{(2)}_{\var}(s_1,s_2,z_i,z_j)$ is such that
$$
\cfrac{\var^{1-d}}{n^2}\Som_{i,j=1}^n\|R^{ij(3)}_{\var}\|_{L^\infty [0,T]\times[0,T]}\leq K_{R_3}(\var^{\zeta}+o(a_n^{-d})).
$$

We recall that,  using Lemma \ref{lem:ind2} and since $q\in \Scr{S}(\R^d)$, we have, $\Scr{P}$-almost everywhere w.r.t. sequences of initial data $Z_\infty$ and for any $D\subset \R^d$,  the bounds
\begin{equation}
\label{auff}
\cfrac{\var^{1-d}}{n^2}\Som_{i,j=1}^n \|R^{(3)}_{\var}\|_{L^\infty [0,T]\times[0,T]}\int_{D}dc\, Q_n^{ij}(s,t,c)\leq K_q \Theta a_n^d \|q\|_{\infty}(\var^{\zeta} + o(a_n^{-d})),
\end{equation}

$$
\frac{1}{n^2}\sum_{i,j=1}^{n}
\int_{{\Cal T}_{\var} (t,z_i)}
dc \, Q_n^{ij}(s,t,c)
\leq C_d t \var^{d-1} (a_n^{d} \|q\|_{\infty})\frac {K_1}n  \sum_{i=1}^{n}|v_i|.
$$

Therefore, using the definition of $\Cal{R}_{\var}^{ij\,(k-2)}$, given in (\ref{rhodef}), we get the equality
$$
[\Lambda_n]_L{\E^n}\left[\left|{\hat A}^{(k)}_{n,\var,M} (t,c)\right|^2 I_{\{M\geq 1\}}\right]= R^{(4)}_{\var}(t) +
$$

$$
\frac1{n^2}\sum_{i,j =1}^n \int_{[0,t]^2}\!\!\! ds_1 ds_2 e^{-\mu [(\rho^{(k-2)}(s_1,z_i)+\rho^{(k-2)}(s_2,z_j)(1-\de_{ij})]}\!\!\!
\int_{\Lambda_n}
\!\!\!
dc\,Q_n^{ij}(s_1,s_2,c)
$$
where
$$
\Sup_{t\in[0,T]}R^{(4)}_{\var}(t)\leq 4 K_q \Theta T^2 \|q\|_{\infty} a_n^{d}(\var^{\zeta} + o(a_n^{-d}))+ O(\var^{d-1}a_n^{d})\to 0,
$$
since $\Lim_{n \to \infty }a_n^d \var^{\zeta}=0$ with $\zeta\in (0, \frac12 -\frac1{2d})$.

Let $z=(x,v)$ and $\tilde{z}=(y,w)$. Define the sequence of measurable sets $\Lambda_n^s[z]=\{a\in\R^d : \exists b\in \Lambda_n\,\, s.t.\,\,a=b-x(s)\}$ and the sequence of functions $Q_{\Lambda_n^s}^*(z,x)=(I_{\Lambda_n^s[z]}q_n)\ast q_n (x)$.

Since $q$ is a radial function, we have the following equality:
\begin{eqnarray*}
\lefteqn{\frac{1}{n^2}\sum_{i,j=1}^{n}e^{-\mu [(\rho^{(k-2)}(s_1,z_i)+\rho^{(k-2)}(s_2,z_j)(1-\de_{ij})]}\int_{\Lambda_n} dc\, Q_n^{ij}(s_1,s_2,c)
=}\\
\\
&\!\!\!\!\!\int_{(\R^d\times\R^d)^2} dz d\tilde{z}
[\mu_n^0
\otimes\mu_n^0](z,\tilde{z})e^{-\mu[\rho^{(k-2)}(s_1,z)+\rho^{(k-2)}(s_2,\tilde{z})]}
Q_{\Lambda_n^{s_2}}^*(\tilde{z},
x(s_1)-y(s_2))&\\
\\
&\!\!\!\!\!+ \cfrac{1}{n} \int_{\R^d} dz \mu_n^0 (z)e^{-\mu\rho^{(k-2)}(s_1,z)}[1
-e^{-\mu\rho^{(k-2)}(s_2,z)}]Q_{\Lambda_n^{s_2}}^*(\tilde{z},
x(s_1)-y(s_2)),
&
\end{eqnarray*}
where the last term is bounded by
$$\frac 2n \int dz\mu_n^0 (q_n\ast q_n)(x(s_1)-y(s_2))\leq 2\frac{a_n^{d}}n\|q\|_\infty\Theta.$$
In the first term we can use Fubini's theorem, rewriting
$$
\int_{(\R^d\times\R^d)^2} dz d\tilde{z}
[\mu_n^0
\otimes\mu_n^0](z,\tilde{z})e^{-\mu[\rho^{(k-2)}(s_1,z)+\rho^{(k-2)}(s_2,\tilde{z})]}
Q_{\Lambda_n^{s_2}}^*(\tilde{z},
x(s_1)-y(s_2))=
$$
$$
\!\!\!\int_{\R^d} da \int_{(\R^d\times\R^d)^2}
dz d\tilde{z}
[\mu_n^0
\otimes\mu_n^0](z,\tilde{z})
\left[e^{-\mu \rho^{(k-2)}(s_1,z)}q_n(a-x(s_1))\right]\times
$$
$$
\times
\left[e^{-\mu \rho^{(k-2)}(s_2,\tilde{z})}
q_n(a-y(s_2))I_{\Lambda_n^{s_2}[\tilde{z}]}(a-y(s_2))\right]
$$
and, because of (\ref{somma1}), we can use Lebesgue's dominated convergence theorem to pass to the limit into the integral.

$\{\Lambda_n\}_{n=1}^\infty$ is an increasing sequence of Lebesgue measurable set and $I_{\Lambda_n}\to 1$
in $L^{\infty}(\R^d)$.
From the $\Scr{P}$-a.e. weak convergence $\otimes_{k=1}^{Q} q_n\mu^0_n \deb \otimes_{k=1}^Q \de_0 f_0 $ and the $L^{\infty}$ convergence of $I_{\Lambda_n^s[y]}$,
we have
$$
\begin{array}{l}
\int_{[0,t]^2 \times (\R^d \times\!\R^d)^2} ds_1 ds_2 dz d\tilde{z}
[\mu_n^0
\otimes \mu_n^0](z,\tilde{z})\times\\
\\
\times\; e^{-\mu[\rho^{(k-2)}(s_1,z)+\rho^{(k-2)}(s_2,\tilde{z})]}
Q_{\Lambda_n^{s_2}}^*\!(\tilde{z}, x(s_1)\!-\!y(s_2))\\
\\
\to \int_{\R^d}
dc \,\left|\bar{V}^{(k)}(t,c)\right|^2,
\end{array}
$$
so that we obtain finally
$$
\Lim_{n\to\infty}
\frac{1}{n^2}\!\Som_{i,j=1}^{n}
\int_{[0,t]^2\times\Lambda_n} ds_1 \, ds_2 \, dc\,
Q_n^{ij}
e^{-\mu [\rho^{(k-2)}(s_2,z_i)+\rho^{(k-2)}(s_1,z_j)(1-\de_{ij})]}=
$$
$$
\int_{\R^d}\!\!\!
dc\left|\bar{V}^{(k)}(t,c)\right|^2\!\!\!.
$$

We compute now the double product and, because of the definitions of $\rho^{(k)}_\var$, $\Lambda_n$ and of the boundedness of $\bar{V}^{(k)}$, we get :
$$
\begin{array}{l}
[\Lambda_n]_L{\E^n}\left[{\hat A}^{(k)}_{n,\var,M} (t,c) \bar{V}^{(k)}(t,c)I_{\{M\geq 1\}}\right]= O(a_n^d \var^{d-1})+\\
\\
\frac 1n\Som_{j=1}^{n}\int_0^t ds_1 \left(
\exp\left[-\mu_{\var}\int_{{\Cal T}_\var (s_1,z_j)} dc\, e^{- {\bar V}^{(k-2)}(T_{z_j, c},c)}\right]\right.\\
\\
\left.\times
\int_{\Lambda_n}
dc \bar{V}^{(k)} (t,c)
q_n (x_j (s_1) -c) \right).
\end{array}
$$
Using (\ref{camvar}), $\Scr{P}$-a.e. with respect to $Z_\infty$, we get:
$$
[\Lambda_n]_L{\E^n}\left[{\hat A}^{(k)}_{n,\var,M} (t,c) \bar{V}^{(k)}(t,c)I_{\{M\geq 1\}}\right]=
$$

$$
\frac 1n \Som_{i=1}^{n}\int_0^t ds_1
e^{-\mu \rho^{(k-2)}(s_1,z_i)}
\int_{\Lambda_n} \, dc\,\bar{V}^{(k)} (t,c) q_n (x_i (s_1) -c)+o(\var)
$$
so that, because of the $\Scr{P}$-a.e. weak convergence of $q_n\mu_n^0$ and the hypothesis on $f_0$, we obtain:
$$
-2\Lim_{n\to\infty}[\Lambda_n]_L{\E^n}\left[{\hat A}^{(k)}_{n,\var,M} (t,c) \bar{V}^{(k)}(t,c)I_{\{M\geq 1\}}\right]
=-2\int_{\R^d}
dc \,\left|\bar{V}^{(k)}(t,c)\right|^2 .
$$
For the last term we have:
$$
[\Lambda_n]_L{\E^n}\left[|\bar{V}^{(k)}(t,c)|^2 I_{\{M\geq 1\}}\right]=(1-e^{-\mu_\var [\Lambda_n]_L})\int_{\Lambda_n}dc\left|\bar{V}^{(k)}(t,c)\right|^2,
$$
whose limit is $\int_{\R^d}dc\left|\bar{V}^{(k)}(t,c)\right|^2$ because of the $L^{\infty}$ convergence of $I_{\Lambda_n}$.

Hence, we have established Proposition \ref{lem:ind}.
\end{proof}

In addition to Proposition  \ref{lem:ind}, we shall need its corollary

\begin{corol}
\label{equiostcomp}
Under the same hypothesis as Proposition \ref{lem:ind}, for the obstacles life functions defined in (\ref{vitak}) and in (\ref{cap}) and for all $\phi\in C_K(\R^d)$, $\Scr{P}$-almost everywhere w.r.t. sequences of initial data $Z_\infty$:
$$
\Lim_{n\rightarrow \infty}\E^n [\var^{d-1} \sum_{i=1}^{M}I_{\mathrm{supp}\phi}(c_i)|\hat{\eta}_{n,\var,M}(t, c_i)-\bar{\eta}(t,c_i) |]=0.
$$
\end{corol}

\begin{proof}
Since for all $\de>0$ (cf. Remark \ref{oss})
$$
\E^n [\var^{d-1} \sum_{i=1}^{M}I_{\mathrm{supp}\phi}(c_i)|\hat{\eta}_{n,\var,M}(t, c_i)-\bar{\eta}(t,c_i) |]\leq
$$
$$
\frac 1\de \E^n[\var^{d-1} M I_{\mathrm{supp}\phi}(c)|{\bar V}^{(k)}(t,c) -{\hat A}^{(k)}_{n,\var,M}(t,c)|] +2\de\mu\left| \mathrm{supp} \phi \cap \Lambda_n\right|,
$$
the proof follows easily bounding the first term on the right-hand side by Cauchy-Schwarz's inequality and using the identity
$$
\E^n[(\var^{d-1} M)^2 I_{\mathrm{supp}\phi}(c)]=\mu^2 [\Lambda_n]_L [\mathrm{supp}\phi \cap \Lambda_n]_L
$$
and Proposition \ref{lem:ind}.
\end{proof}

We shall need moreover a modified form of the proposition, whose proof will be only sketched, being essentially the same as the one of Proposition \ref{lem:ind}

\begin{prop}
\label{corrie}

Under the same hypothesis as Proposition \ref{lem:ind}, $\!\!\Scr{P}$-almost everywhere w.r.t. sequences of initial data $Z_\infty$:
$$
\Lim_{n\rightarrow \infty}\Sup_{1\leq h\leq n}\left(\frac 1{|v_h|}{\E^n}\left[
|{\hat A}^{(k)}_{n,\var,M}
-\bar{V}^{(k)}|^2 (T_{z_h,c_1},c_1)M I_{c_1 \in {\Cal T}_\var (T,z_h)}\right]\right)=0.
$$

\end{prop}

\begin{proof}

The proof is obtained along the same line as the proof of Prop. \ref{lem:ind}, since $T_{z_h, c_1}$ is independent of the stochastic times $\mbox{\Large{$\tau$}}_{\infty}$, is bounded by $T$ and all quantities involved depend on $T_{z_h, c_1}$ in a simple way, so to allow to get estimates uniform in $z_h$.

Using, instead of (\ref{auff}), the estimates
$$
\cfrac{\var^{1-d}}{n^2}\Som_{i,j=1}^n \|R^{ij(3)}_{\var}\|_{L^\infty [0,T]\times[0,T]}Q_n^{ij}(s,t,c)\leq K_r a_n^d (\var^{\zeta} + o(a_n^{-d})),
$$
$$
\frac{1}{n^2}\sum_{i,j=1}^{n}
\int_{{\Cal T}_{\var} (t,z_h)\cap{\Cal T}_{\var} (t,z_i)}
dc\,  Q_n^{ij}(s,t,c)
\leq K_Q |v_h|\var^{d-1} a_n^{d} (\var^{\zeta} + o (a_n^{-d}))
$$
which are a consequence of Lemma \ref{lem:ind2} and of the bounds (\ref{somma1}), (\ref{somma2}), we perform the change of variables (\ref{camvar}) and we get:
$$
\frac1{|v_h|}{\E^n}\left[M\left|{\hat A}^{(k)}_{n,\var,M} (T_{z_h,c_1},c_1)\right|^2 I_{c_1 \in {\Cal T}_\var (T,z_h)}\right]=O( a_n^d (\var^{\zeta} + o(a_n^{-d})))+
$$

$$
\frac {\mu C_d }{n^2} \sum_{i,j =1}^n \int_0^T ds \int_{[0,s]^2} \!\!\!ds_1
ds_2
Q_n^{ij}(s_1,s_2,x_h(s))
e^{-\mu [\rho^{(k-2)}(s_1,z_i)+\rho^{(k-2)}(s_2,z_j)]},
$$
while for the double product, we have
$$
\frac1{|v_h|}{\E^n}\left[M I_{\{c_1\in {\Cal T}_\var (T,z_h) \}}{\hat A}^{(k)}_{n,\var,M} (t,c_1) \bar{V}^{(k)}(t,c_1)\right]=
O(a_n^{d+1}\var)+
$$

$$
\frac {\mu C_d}n\Som_{j=1}^{n}\!\!\int_0^T\!\!\! ds \bar{V}^{(k)} (s,x_h(s))\!\int_0^{s}\!\!\!  ds_1 q_n (x_j (s_1) -x_h(s))
\exp[-\mu \rho^{(k-2)}(s_1,z_j) ].
$$
Of course
$$
\frac1{|v_h|}{\E^n}\left[M\left|\bar{V}^{(k)}(t,c_1)\right|^2 I_{c_1 \in {\Cal T}_\var (T,z_h)}\right]\!=\! \mu C_d \int_0^T ds \left|\bar{V}^{(k)}(s,x_h(s))\right|^2 + O(\var).
$$
Then, observing that, thanks to the $\Scr{P}$-a.e. weak convergence of $q_n\mu_n^0$,  the quantity  $$\frac 1{|v_h|}{\E^n}\left[
|{\hat A}^{(k)}_{n,\var,M} (T_{z_h,c_1},c_1)-\bar{V}^{(k)}(T_{z_h,c_1},c_1)|^2 M I_{c_1 \in {\Cal T}_\var (T,z_h)}\right]$$
vanishes for all fixed $h\in\N$, uniformly in $z_h$ because of the choice of the initial data (cf. Remark \ref{rdebotrasl} in Appendix \ref{app2}) and since all error estimates in this section are themselves uniform in $z_h$, the Proposition is proved.
\end{proof}

\subsection{Equivalence between (\ref{cap}) and  (\ref{approxlevk}): the $\cO,\cP$-frozen system}\label{frozen}

The last step to complete the proof of our main theorem is to  estimate the distance between the particle system (\ref{cap}) and the particle system
(\ref{approxlevk}). The difficulty, here, originates from the complicate dependence of the life functions, defining both systems, on the stochastic configuration: in particular, this dependence compels us to use (\ref{rischio}) for estimating differences of life functions for both species, preventing the use of any of its simplified forms, mentioned in Remark \ref{oss}.

Since the system (\ref{cap}) is equivalent in the large particles number limit to the system (\ref{sistk}), a mean--field system, what should happen for systems (\ref{cap}) and (\ref{approxlevk}) to be equivalent in the limit is that (\ref{approxlevk}) is somehow stable (in the large particles number limit) w.r.t. deletion (or addition) of a finite number of particles of the two species. This is in fact the case, as we shall show in the first Lemma in this section. Then, deleting light particles and obstacles in a suitable way from system (\ref{approxlevk}), we shall be able to simplify the dependence on the stochastic configuration in the estimates and finally to prove the asymptotic equivalence between (\ref{cap}) and (\ref{approxlevk}).

\medskip

Let $\cO$ and $\cP$ be finite subsets of ${\mathbb N}$ and $k$ the level of approximation. In this section, we introduce a sequence of particle systems where the obstacles with labels in $\cO$ (resp. light particles with labels in $\cP$) have no interaction with the light particles (resp. with the obstacles). These systems will give us a tool to estimate the required distance.

For a given configuration of obstacles, $\mathbf{c}_M=(c_1, \ldots, c_M)$ and a particle initial datum $\mathbf{z}_n=(z_1, \ldots, z_n) $, fix the sets of integers $\cO \subseteq \{1,\ldots, M\}$, $\cP \subseteq \{1, \ldots, n\}$ and define, for integers $M,k\geq1$:
\begin{equation}
\label{frozen-system}
\begin{array}{l}
\gamma_{n, \varepsilon, 0}^{(0,\emptyset,\cP)}(t,z_j)=\gamma_{n, \varepsilon, 0}^{(k,\emptyset,\cP)}(t,z_j)=I_{\{j \notin \cP\}},\\
\\
\gamma_{n, \varepsilon, M}^{(0,\cO,\cP)}(t,z_j)=I_{\{j \notin \cP\}},\\
\\
\delta_{n,\varepsilon,M}^{(0,\cO,\cP)}(t, c_\ell)=I_{\{\ell \notin \cO\}},\\
\\
\gamma_{n,\varepsilon,M}^{(k,\cO,\cP)}(t,z_j)= I_{\{x_j (s) \notin \Bigcup_{h=1}^M B_\var(c_h)\delta^{(k-1,\cO,\cP)} (c_h, s) \; \forall s \in [0,t)\}}I_{\{j \notin \cP\}},\\
\\
\delta_{n,\varepsilon,M}^{(k,\cO,\cP)}(t,c_\ell)=I_{\{A_{n,\varepsilon,M}^{(k,\cO,\cP)}(t, c_\ell) \leq \tau_{c_\ell}\}} I_{\{\ell \notin \cO\}},\\
\\
A_{n,\varepsilon,M}^{(k,\cO,\cP)}(t,c_\ell)=\cfrac{1}{n}\Som_{i=1}^n \int_{0}^{t} q_n (x_i(s)-c_\ell)\gamma_{n,\varepsilon,M}^{(k-1,\cO,\cP)}(s,z_i) ds .
\end{array}
\end{equation}

Remark that $A_{n,\varepsilon,M}^{(k,\cO,\cP)}(t,c_\ell)$ does not depend on $\tau_{c_h}$, $x_j, v_j$ for any $h \in \cO$ and any $j\in \cP$, and it does not depend on $c_h$, $h\in\cO$, whenever $\ell\neq h$. Moreover, for $(\cO,\cP)=(\emptyset, \emptyset)$, the system defined above is the $k$-th level approximation system defined in (\ref{approxlevk}), Section \ref{approximation}, i.e.
$$A_{n,\varepsilon,M}^{(k,\emptyset, \emptyset)}=V^{(k)}_{n,\varepsilon,M}, \qquad\gamma_{n,\varepsilon,M}^{(k,\emptyset, \emptyset)}=\xi_{n,\varepsilon,M}^{(k)}, \qquad \delta_{n,\varepsilon,M}^{(k,\emptyset, \emptyset)}=\eta_{n,\varepsilon,M}^{(k)}.$$

\smallskip

\subsubsection{Proof of the equivalence between (\ref{cap}) and  (\ref{approxlevk}) using the $\cO,\cP$-frozen system}

Using the bound (\ref{riduzione}), we can use the same strategy as in \cite{NOR}.

To shorten the notation, we define in this paragraph, for sets $\cO,\cO_1\subseteq \{1,\ldots, M\}$ and $\cP, \cP_1 \subseteq \{1, \ldots, n\}$:
$$
\Cal{A}^{nM(k)}_{\cO,\cP}(\cO_1,\cP_1)=\left| A_{n,\var,M}^{(k,\cO\cup\cO_1,\cP\cup\cP_1)}
-  A_{n,\var,M}^{(k,\cO,\cP)}\right|,
$$

\noindent
for $M= 1, \ldots$,

\begin{eqnarray*}
\Cal{I}_{\var,t,z_i}^{\cO, \cP, 0}(c_\ell)&=& 0\\
\\
\Cal{I}_{\var,t,z_i}^{\cO, \cP, M}(c_\ell)&=& M  I_{\{c_\ell\in {\Cal T}_\var (t,z_i)\}} I_{\cO\subset\{1,\ldots,M\}}I_{\cP\subset\{1,\ldots,n\}},
\end{eqnarray*}
and $\Cal{I}_{\var,t,z_i}^{M}=\Cal{I}_{\var,t,z_i}^{\emptyset, \emptyset, M}$.

We prove the following lemmas:

\begin{lemma}
\label{auxil}
Consider the stochastic variables defined in (\ref{frozen-system}). Then,
under the same hypothesis as in Proposition \ref{lem:ind}, for all integers $k\geq 1$, $p\geq 1$ and for all ${\Cal O}_1, {\Cal P}_1\subset \N$ s.t. $0<\#{\Cal O}_1, \#{\Cal P}_1 < \infty$, $\Scr{P}$-almost everywhere w.r.t. sequences of initial data $Z_\infty$, we have the limits

\begin{equation}
\label{aux1}\\
\Lim_{n\to\infty} \frac{1}{n^p}
\sum_{\lmul{1\leq j_1\leq n\\,\ldots,\\ 1\leq j_p \leq n}
}
\Sup_{z_{j_0}\in\R^d\times\R^d}\frac 1{|v_{j_0}|}\E^n\!\!\left[\int_{[0,T]^p}
\prod_{h=1}^{p}ds_{j_h} q_n (x_{j_h} (s_{j_h}) -c_{i_h}) I_{\{c_{i_h} \in {\Cal T}_\var (s_{j_{h-1}},z_{j_{h-1}})\}}\right.
\end{equation}
$$
\left.\times
\Cal{I}_{\var,s_{j_p}, z_{j_p}}^{\cO\cup\cO_1,\cP\cup\cP_1, M}(c_\ell)
M^{\#\{i_1,\ldots, i_p, \ell\}-1}
\Cal{A}^{nM(k)}_{\cO,\cP}(\cO_1,\cP_1)(T_{z_{j_p},c_\ell},c_\ell)
\phantom{\int}\right] =0
$$

$$
\Lim_{n\to\infty}\Sup_{z_{j_0}\in\R^d\times\R^d}\frac 1{|v_{j_0}|}\E^n\left[\Cal{I}_{\var,s_{j_0}, z_{j_0}}^{\cO\cup\cO_1,\cP\cup\cP_1, M}(c_\ell)\Cal{A}^{nM(k)}_{\cO,\cP}(\cO_1,\cP_1)(T_{z_{j_0},c_\ell},c_\ell)\right] =0\phantom{\int_{[0,T]^p}
\prod_{h=1}^{p}ds_{j_h} q_n (x_{j_h} (s_{j_h}) -c_{i_h}) I_{\{c_{i_h} \in {\Cal T}_\var (s_{j_{h-1}},z_{j_{h-1}})\}}}
$$

\noindent
for all $i_1 ,\ldots, i_p \in  \Cal O \cup {\Cal O}_1 $ and $s_{j_0}<T$, $z_{j_0}=(x_{j_0}, v_{j_0})$.
Limits (\ref{aux1}) are valid also in the case where $T$ replaces $T_{z_j,c_\ell}$.

\end{lemma}

\begin{rem}
This lemma shows that particle systems obtained from (\ref{approxlevk}) by deleting a finite number of light particles and/or obstacles are equivalent, in the prescribed asymptotics.
\end{rem}

\begin{proof}

The proof is obtained by induction.

For $k=1$ and $p\geq 0$, we have
$$
\Cal{A}^{nM(1)}_{\cO,\cP}(\cO_1,\cP_1)(T_{z_{j_p},c_\ell},c_\ell)\leq
\frac 1n \sum_{h\in \cP_1}\int_0^{T_{z_{j_p}, c_\ell}} ds q_n (x_h(s)-c_\ell).
$$
and the second limit in (\ref{aux1}) follows  straightforwardly, since
$$\Cal{I}_{\var,s_{j_p}, z_{j_p}}^{\cO\cup\cO_1,\cP\cup\cP_1, M}(c_\ell)\Cal{A}^{nM(1)}_{\cO,\cP}(\cO_1,\cP_1)(T_{z_{j_p},c_\ell},c_\ell)\leq
$$
$$ \Cal{I}_{\var,s_{j_p}, z_{j_p}}^{\cO\cup\cO_1,\cP\cup\cP_1, M}(c_\ell) \frac{a_n^d}{n}\#\cP_1 \|q\|_{\infty}T$$
and $\E^n[\Cal{I}_{\var,s_{j_p}, z_{j_p}}^{\cO\cup\cO_1,\cP\cup\cP_1, M}(c_\ell)]\leq \mu C_d T |v_{j_p}|$, for $p\geq 0$.

Denoting
$
N=\#\{i_1,\ldots, i_p, \ell\}
$, we build a partition of labels by grouping them in the following way.

We start from the label $\ell$ and we denote as $i_{k^{(1)}_1},\ldots, i_{k^{(1)}_{w_1}}$, with $k^{(1)}_1<k^{(1)}_2<\ldots<k^{(1)}_{w_1}$, the $w_1$ labels among the $p+1$ labels $i_1,\ldots, i_p, \ell$ having common value $\ell$; notice that $k^{(1)}_{w_1}=p+1$, being $c_{i_{p+1}}=c_{\ell}$ the obstacle associated to the light particle label $j_{p+1}$.
We then consider $i_{f_1}=\max\{i_s : i_s \neq \ell\}$ and we call  $i_{k^{(2)}_{1}},\ldots, i_{k^{(2)}_{w_2}}$ the $w_2$
labels having common value $i_{f_1}$, always using the ordering $k^{(2)}_1<k^{(2)}_2<\ldots<k^{(2)}_{w_2}$  (i.e., $k^{(2)}_{w_2}=f_1$). We build in this way $N$ groups of labels and we denote $\bar{q}$ the group label such that $k^{(\bar{q})}_{1}=1$ .

We can write then
\begin{equation}
\label{uffaux}\\
\frac{1}{n^p}
\sum_{\lmul{1\leq j_1\leq n\\,\ldots,\\ 1\leq j_p \leq n}
}
\E^n\!\!\left[\int_{[0,T]^p}
\prod_{h=1}^{p}ds_{j_h} q_n (x_{j_h} (s_{j_h}) -c_{i_h}) I_{\{c_{i_h} \in {\Cal T}_\var (s_{j_{h-1}},z_{j_{h-1}})\}}\right.
\end{equation}
$$
\left.\times
\Cal{I}_{\var,s_{j_p}, z_{j_p}}^{\cO\cup\cO_1,\cP\cup\cP_1, M}(c_\ell)
M^{\#\{i_1,\ldots, i_p, \ell\}-1}
\frac 1n \sum_{j_{p+1}\in \cP_1}\int_0^{T_{z_{j_{p}}, c_\ell}} ds q_n (x_{j_{p+1}}(s)-c_\ell)
\phantom{\int}\right]\leq
$$
$$
\frac{\E^n_c[M^N]}{[\Lambda_n]^N}\frac{1}{n^{p+1}}\sum_{j_{k^{(1)}_{w_1}}\in \cP_1}
\sum_{\lmul{1\leq j_{k^{(1)}_{h}}\leq n\\
h=1,\ldots,w_1-1
}}
\sum_{\lmul{1\leq j_{k^{(2)}_{h}}\leq n\\h=1,\ldots,w_q}}
\ldots
\sum_{\lmul{1\leq j_{k^{(N)}_{h}}\leq n\\h=1,\ldots,w_N}}
\prod\limits_{q=1}^{N}
\int_{\Lambda_n} dc_{i_{k^{(q)}_{w_{q}}}}\int_{[0,T]^{w_{q}}}
$$
$$
\prod\limits_{h=1}^{w_{q}}
\left[ds_{j_{k^{(q)}_{h}}}
q_n
(x_{j_{k^{(q)}_{h}}}(s_{j_{k^{(q)}_{h}}})-c_{i_{k^{(q)}_{w_{q}}}})
I_{\{c_{i_{k^{(q)}_{w_{q}}}}\in \bigcap\limits_{h=1}^{w_{q}}
{\Cal T}_\var (s_{j_{k^{(q)}_{h}-1}},z_{j_{k^{(q)}_{h}-1}})\}}
\right]
$$

We use then the bound
$$
I_{\{c_{i_{k^{(q)}_{w_{q}}}}\in \bigcap\limits_{h=1}^{w_{q}}
{\Cal T}_\var (s_{j_{k^{(q)}_{h}-1}},z_{j_{k^{(q)}_{h}-1}})\}}
\leq I_{\{c_{i_{k^{(q)}_{w_{q}}}}\in
{\Cal T}_\var (s_{j_{k^{(q)}_{1}-1}},z_{j_{k^{(q)}_{1}-1}})\}}.
$$
so that, thanks to (\ref{somma1}), we have:
$$
\!\!\!(\ref{uffaux})\leq
\frac{\E^n_c[M^N]}{[\Lambda_n]^N}
\left(\#\cP_1 \|q\|_{\infty}T\frac{a_n^d}{n}\right)
\frac{\prod\limits_{q=1}^{N}(K_1 T)^{w_q-1}}{n^{N-1}}
\int_{\Lambda^N_n}
\prod\limits_{q=1}^{N}
dc_{i_{k^{(q)}_{w_{q}}}}\left\{
I_{\{c_{i_{k^{(\bar{q})}_{w_{\bar{q}}}}}\in
{\Cal T}_\var (s_{j_{0}},z_{j_{0}})\}}
\right.
$$
$$
\left.
\!\!\!\!\!\!\!\!\!\prod\limits_{\lmul{q=1\\q\neq \bar{q}}}^{N}
\left[\sum_{1\leq j_{k^{(q)}_{1}}\leq n}
\int_{[0,T]}ds_{j_{k^{(q)}_{1}-1}}
q_n
(x_{j_{k^{(q)}_{1}-1}}(s_{j_{k^{(q)}_{1}-1}})-c_{i_{k^{(q)}_{1}-1}})
I_{\{c_{i_{k^{(q)}_{w_{q}}}}\in
{\Cal T}_\var (s_{j_{k^{(q)}_{1}-1}},z_{j_{k^{(q)}_{1}-1}})\}}
\right]\right\}
$$
and, using (\ref{somma2}) and $\E_c^n[M^{N}]\leq K (\mu_{\var}[\Lambda_n])^{N}$
we get finally:
$$
(\ref{uffaux})
\leq K \left(\#\cP_1 \|q\|_{\infty}T\frac{a_n^d}{n}\right)
\left( \frac{\mu T C_d K_2}{K_1}\right)^{N}
\frac{(T K_1)^{p+1}}{K_2}|v_{j_0}|
$$
so that (\ref{aux1}) is valid for $k=1$.

For $k=2$, since
$$
\Cal{I}_{\var,s_{j_p}, z_{j_p}}^{\cO\cup\cO_1,\cP\cup\cP_1, M}(c_\ell)\Cal{A}^{nM(2)}_{\cO,\cP}(\cO_1,\cP_1)(T_{z_{j_p},c_\ell},c_\ell)\leq\phantom{\Cal{I}_{\var,s_{j_p}, z_{j_p}}^{\cO\cup\cO_1,\cP\cup\cP_1, M}(c_\ell)}
$$
$$
\Cal{I}_{\var,s_{j_p}, z_{j_p}}^{\cO\cup\cO_1,\cP\cup\cP_1, M}(c_\ell)\left[
\frac 1n \sum_{h\in \cP_1}\int_0^{T_{z_{j_p}, c_\ell}} ds q_n (x_h(s)-c_\ell)
+
\right.
$$
$$
\left.
\frac 1n \sum_{h\notin \cP_1}\int_0^{T_{z_{j_p}, c_\ell}} ds q_n (x_h(s)-c_\ell)
\Som_{m\in\cO}I_{\{c_{m} \in {\Cal T}_\var (s_h,z_h)\}}
\right],
$$
in the same way (using Lemma \ref{lemmatubo} and Lemma \ref{misI}) we get the second limit in (\ref{aux1}) and
$$
(\ref{uffaux})
\leq K_D\left(\#\cP_1 \frac{a_n^d}{n}+\#\cO(a_n^d\var^{\zeta}+\var^{d-1})\right)
\left( \frac{\mu T C_d K_2}{K_1}\right)^{N}
\frac{(T K_1)^{p+1}}{K_2}|v_{j_0}|
$$
(with $K_D$ depending on $\|q\|_{\infty}$, $T$, $K_1$, $K_2$, $C_d$).

Let assume then that (\ref{aux1}) is true for $k-1$, with $k>2$. From the definition of
$A_{n,\var,M}^{(k,\cO,\cP)} (T_{z_{j_p},c_\ell},c_\ell)$ and the bound (\ref{riduzione}) we have, for $p\geq 0$:
\begin{eqnarray}
\label{ricor1}\\
\nonumber
\lefteqn{\Cal{I}_{\var,s_{j_p}, z_{j_p}}^{\cO\cup\cO_1,\cP\cup\cP_1, M}(c_\ell)
\Cal{A}^{nM(k)}_{\cO,\cP}(\cO_1,\cP_1)(T_{z_{j_{p}},c_\ell},c_\ell)
\leq \Cal{I}_{\var,s_{j_p}, z_{j_p}}^{\cO\cup\cO_1,\cP\cup\cP_1, M}(c_\ell)\times}\\
\nonumber
&&\left\{\left(\#\cP_1 \|q\|_{\infty}T\frac{a_n^d}{n}\right)+\frac 1n \sum_{j_{p+1}\notin \cP} \int_0^{T} ds q_n (x_{j_{p+1}} (s) -c_{\ell})\times\right.\\
\nonumber
&&
\left.\left[\sum_{m=1}^{M}
I_{\{c_m \in {\Cal T}_\var (s,z_{j_{p+1}})\}}
|\de_{n,\var,M}^{(k-2,\cO,\cP)}
- \de_{n,\var,M}^{(k-2,\cO\cup {\cO}_1,\cP\cup {\cP}_1)}|(T_{z_{j_{p+1}},c_m},c_m)\right]\right\},
\end{eqnarray}
and from Lemma \ref{3.2}, for $\de>0$,
\begin{equation}
\label{ricor2}
\E_{\tau_m}\left[\left| \de^{(k-2,\cO\cup {\cO}_1,\cP\cup {\cP}_1)}_{n,\var,M}
-   \delta^{(k-2,\cO,\cP)}_{n,\var,M}\right| (T_{z_{j_{p+1}},c_m},c_m)\right]
\leq 2\de+
\end{equation}

$$
\frac 1\delta
{\E_{\tau_m}}\left[
\Cal{A}^{nM(k-2)}_{\cO,\cP}
(\cO_1\cup\{m\},\cP_1\cup\{j_{p+1}\})(T_{z_{j_{p+1}},c_m},c_m)\right]
$$

$$
+\frac 1\delta
{\E_{\tau_m}}\left[\Cal{A}^{nM(k-2)}_{\cO\cup\cO_1,\cP\cup\cP_1}
(\{m\},\{j_{p+1}\})
(T_{z_{j_{p+1}},c_m},c_m)\right].
$$
When substituting (\ref{ricor2}) through (\ref{ricor1}) into the expectation values in formula (\ref{aux1}), the contribution to the expectation value coming from the first term on the righthand side of (\ref{ricor2}) is bounded by $2 K_I \de$, where $K_I=K_I (T, C_d, \mu,K_1,K_2,p,\#\{i_1,\ldots, i_p, \ell\})$, while the contribution to (\ref{aux1}) due to the last terms in (\ref{ricor2}) vanish asymptotically because of the inductive hypothesis.  The remaining term coming from (\ref{ricor1}) is $O(a_n^d n^{-1})$.

The last part of the Lemma is proved
replacing $T_{z_j,c_\ell}$ by $T$.
\end{proof}

\begin{lemma}

\label{gelato}
Consider the stochastic variables defined in (\ref{cap}) and (\ref{frozen-system}).
Under the same hypothesis as in Proposition \ref{lem:ind},
$\forall k\geq 1$, $\forall \cO, \cP \subset \N$ such that \mbox{$0 \le \#\cO, \#\cP < \infty$}, $\Scr{P}$-almost everywhere w.r.t. sequences of initial data $Z_\infty$,
\begin{equation}
\label{teorgelato}
\Lim_{n\to\infty}
\Sup_{1\leq u \leq n}\frac 1{|v_u|} \E^n\big[\Cal{I}_{\var,T,z_u}^{\cO, \cP, M}(c_\ell)|A^{(k, \cO, \cP)}_{n,\var, M} - {\hat A}^{(k)}_{n,\var, M}| (T_{z_u,c_\ell},c_\ell) |\big]=0.
\end{equation}

\end{lemma}

\begin{proof}

We can prove Lemma \ref{gelato} by induction.

For $k=1$ and $1\leq u \leq n$ we have:
$$
\!\!\frac 1{|v_u|}\E^n[\Cal{I}_{\var,T,z_u}^{\cO, \cP, M}(c_\ell) |A_n^{(1,\cO,\cP)}
- {\hat A}^{(1)}_n |(T_{z_u,c_\ell},c_\ell)
]
\leq
(\#\cP) C_d T^2  \mu
\|q\|_{\infty} \frac {a_n^{d}}{n},
$$
and for $k=2$ (using Lemma \ref{misI})
$$
\!\!\frac 1{|v_u|}\E^n[\Cal{I}_{\var,T,z_u}^{\cO, \cP, M}(c_\ell) |A_n^{(2,\cO,\cP)}
- {\hat A}^{(2)}_n |(T_{z_u,c_\ell},c_\ell)
]
\leq K_m (\#\cP \frac {a_n^{d}}{n} + \#\cO \var^{d-1} + a_n^{d} \var^{\zeta})
$$
so that (\ref{teorgelato}), because of (\ref{conditio}) and (\ref{conditiovar}), is valid ($K_m$ depends on $\|q\|_{\infty}$, $T$, $\mu$, $K_2$, $C_d$).

\smallskip

Let assume (\ref{teorgelato}) is valid for a given $k>1$ and $\forall \cO, \cP$.

We shall use the notation $\Som_{j\notin\cP}=\Som_{\lmul{1\leq j\leq n\\\\j \notin \cP}}$ and $\Som_{h\notin\cO}=\Som_{\lmul{1\leq h \leq M\\\\ h\notin\cO}}$.

From the definitions (\ref{frozen-system}) and (\ref{cap})
we obtain
\begin{eqnarray*}
\lefteqn{
\frac 1{|v_u|}\E^n [
\Cal{I}_{\var,T,z_u}^{\cO, \cP, M}(c_\ell)|A_{n,\var,M}^{(k,\cO,\cP)}
- {\hat A}^{(k)}_{n,\var,M}| (T_{z_u,c_\ell},c_\ell)
] \leq
(\#\cP) C_d T^2 \mu
\|q\|_{\infty}
\frac {a_n^{d}}{n}+}\\
\\
&&\!
\E^n \Big [\frac{
\Cal{I}_{\var,T,z_u}^{\cO, \cP, M}(c_\ell)} {|v_u|n}
\sum_{j \notin \cP}\int_0^{T}\!\!\! ds  q_n (x_j (s) -c_\ell)
\left|{\hat \xi}^{(k-1)} - \gamma^{(k-1,\cO,\cP)}\right| (s,z_j) \Big] .
\end{eqnarray*}

Since both ${\hat \xi}^{(k-1)}(s,x_j,v_j)$ and $\gamma^{(k-1,\cO,\cP)} (s,x_j,v_j)$ are of the form (\ref{formaxi}), using (\ref{riduzione}) and
then Lemma \ref{lemmatubo} and Lemma \ref{misI},
$\Scr{P}$-almost everywhere w.r.t. sequences of initial data $Z_\infty$
we get (we omit the time integral):
\begin{equation}
\label{gel1}
\frac1{n} \Som_{j \notin \cP} \E^n \left[
\frac{\Cal{I}_{\var,T,z_u}^{\cO, \cP, M}(c_\ell)}{|v_u|}
q_n (x_j (s) -c_\ell) \left|{\hat \xi}^{(k-1)}
- \gamma^{(k-1,\cO,\cP)}\right| (s,z_j)
\right]\leq
\end{equation}
$$
T
C_d \mu\|q\|_{\infty}a_n^d \left[T C_d(\#\cO ) (\frac 1n \Som_{i=1}^{n} |v_i|)\, \var^{d-1}
+ 3 K_c(\var^{\zeta}+o(a_n^{-d}))\right]+
$$
$$
\!\Som_{j \notin \cP}\!
\E^n\!\!\!\left[\frac{\Cal{I}_{\var,T,z_u}^{\cO, \cP, M}(c_\ell)}{n|v_u|}
q_n (x_j (s)\!-\!c_\ell)\!\!\!\!\!
\Som_{h\notin\cO\cup\{\ell\}}\!\!\!\!\! I_{c_h \in {\Cal T}_\var (s,z_j)}
\!\!\left| \bar{\eta}^{(k-2)}\!
-\!\delta^{(k-1,\cO,\cP)}_{n,\var,M}\right|\!\!(T_{z_j,c_h},c_h)\!
\right]\!\!.
$$
We use then lemma \ref{3.2} and we write, for $\delta>0$:
\begin{equation}
\label{ciccia}
\E_{\tau_{\ell}}\left[\left| \bar{\eta}^{(k-2)}
-   \delta^{(k-1,\cO,\cP)}_{n,\var,M}\right| (T_{z_j,c_h},c_h)
\right] \leq 2\de +
\end{equation}
$$
\frac{\E_{\tau_{\ell}}}\delta\left[\left\{\left| \bar{V}^{(k-2)}
- A_{n,\var,M}^{(k-2,\cO\cup \{\ell\},\cP\cup \{j\})}\right|
+
\Cal{A}^{nM(k-2)}_{\cO,\cP}
(\{\ell\},\{j\})
\right\}(T_{z_j,c_h},c_h)
\right]\!,
$$
since $ A_{n,\var,M}^{(k-2,\cO\cup \{\ell\},\cP\cup \{j\})}(T_{z_j,c_h},c_h)$ does not depend on $\tau_{\ell}$.

The contribution to the right-hand side of (\ref{gel1}) coming from the last term in the last expectation value of (\ref{ciccia}) vanishes in the limit thanks to Lemma \ref{auxil} and the contribution coming from the first term is bounded by $2\delta (K_a + K_b a_n^d \var^{\zeta}+o(a_n^{-d}))$ for each $\de>0$.

We evaluate now  the remaining term, recalling that $h\neq \ell$.

When $h\neq \ell$, since $\left| \bar{V}^{(k-2)}
- A_{n,\var,M}^{(k-2,\cO\cup \{\ell\},\cP\cup \{j\})}\right|(T_{z_j,c_h},c_h)$ is independent of $c_\ell$, we have (again omitting the time integral):
\begin{eqnarray*}
\lefteqn{
\!\!\!\!\E^n\!\!
\left[
\frac{\Cal{I}_{\var,T,z_u}^{\cO, \cP, M} (c_\ell)}
{|v_u| n}\!\!\!\!\!\! \Som_{\lmul{j \notin \cP\\ h\notin\cO \cup \{\ell\}}}
\!\!\!
\!\!\!
q_n (x_j (s)\!-\!c_\ell)
I_{c_h \in {\Cal T}_\var (s,z_j)}\!
\left| \bar{V}^{(k-2)}\!\!
-\!\! A_{n,\var,M}^{(k-2,\cO\cup \{\ell\},\cP\cup \{j\})}\right|\!\!(T_{z_j,c_h},c_h)\!
\right]
}
\\
\\
&&
\!\!\!
\leq  \frac {\mu_\var [\Lambda_n]_L}{n |v_u|}
\sum_{j \notin \cP}
\E^n[I_{c\in{\Cal T}_\var (T,z_u)}  q_n (x_j (s)\!-c)]\\
\\
&&\phantom{\leq  \frac {\mu_\var}{n |v_u|}
\sum_{j \notin \cP}}\times
\E^n\left[\Cal{I}_{\var,s,z_j}^{\cO, \cP, M} (c_1) \left| \bar{V}^{(k-2)}
- A_{n,\var,M}^{(k-2,\cO\cup \{\ell\},\cP\cup \{j\})}\right|(T_{z_j,c_1},c_1)
\right].
\end{eqnarray*}

By the triangular inequality, we may then bound the time integral in $[0,T]$ of this quantity,
$\Scr{P}$-almost everywhere w.r.t. sequences of initial data $Z_\infty$, by the sum of:
\begin{eqnarray*}
\lefteqn{\!\!\!
\frac {\mu_\var [\Lambda_n]_L}{n|v_u|}
\sum_{j \notin \cP} \int_0^{T} ds\,
\E^n[I_{\{c\in{\Cal T}_\var (T,z_u)\}} q_n (x_j (s) -c)]}\\
\\
&\times
\E^n\!\!\left[\Cal{I}_{\var,s,z_j}^{\cO, \cP, M} (c_1)\left| \bar{V}^{(k-2)}
- \hat{A}_{n,\var,M}^{(k-2)}\right| (T_{z_j,c_1},c_1)
\right]
&
\\
\\
&\!\!\!\leq
K_V\sqrt{\Sup_{1\leq u \leq n}\!\!\E^n\left[\frac 1{|v_u|}\Cal{I}_{\var,T,z_u}^{M} (c_1)\left| \bar{V}^{(k-2)}
- \hat{A}_{n,\var,M}^{(k-2)}\right|^2 (T_{z_u,c_1},c_1)\right]},&
\end{eqnarray*}
vanishing because of Proposition \ref{corrie}, and
\begin{eqnarray*}
\lefteqn{
\frac {\mu_\var [\Lambda_n]_L}{n |v_u|}
\sum_{j \notin \cP} \int_0^{T}ds\,
\E^n[I_{\{c\in{\Cal T}_\var (T,z_u)\}} q_n (x_j (s) -c)]}\\
\\
&&\times
\E^n\!\!\left[\Cal{I}_{\var,s,z_j}^{\cO, \cP, M} (c_1)
\left|\hat{A}_{n,\var,M}^{(k-2)}
- A_{n,\var,M}^{(k-2,\cO\cup \{\ell\},\cP\cup \{j\})}\right| (T_{z,c_1},c_1)
\right]
\\
\\
&&\leq
K_A\Sup_{1\leq j \leq n}\!\! \E^n\left[\frac{\Cal{I}_{\var,T,z_j}^{\cO, \cP, M} (c_1)}{|v_j|}\left| A_{n,\var,M}^{(k-2,\cO\cup \{\ell\},\cP\cup \{j\})}
- \hat{A}_{n,\var,M}^{(k-2)}\right| (T_{z_j,c_1},c_1)\right]\!,
\end{eqnarray*}
vanishing because of the inductive hypothesis (the constants there depend on $C_d$, $\mu$, $T$, $K_1$, $K_2$).
Since, because of condition $a_n^d \var^{\zeta}\to 0$, all terms vanish, the Lemma is proved.
\end{proof}

Collecting all results in this section we prove

\begin{prop}
\label{capbar}
Consider the stochastic variables defined in (\ref{approxlevk}) and (\ref{cap}). Then,
under the same hypothesis as in Proposition \ref{lem:ind}, $\Scr{P}$-almost everywhere w.r.t. sequences of initial data $Z_\infty$,
\begin{equation}
\label{particelle}
\Lim_{n\to\infty}\frac 1n\sum_{j=1}^{n}\E^n \Big[
\left|{\hat \xi}^{(k)}(t,z_j) - \xi_{n,\varepsilon,M}^{(k)}(t,z_j)\right|\Big]=0
\end{equation}
and
\begin{equation}
\label{ostacoli}
\Lim_{n\to\infty}\E^n \Big[\var^{d-1} M I_{\mathrm{supp}\phi}(c_\ell)|\eta^{(k)}_{n,\var,M}(t, c_\ell)-\hat{\eta}^{(k)}_{n,\var,M}(t,c_\ell) |\Big]=0.
 \end{equation}
\end{prop}

\begin{proof}

We start from (\ref{particelle}), bounding it first as
$$
\frac 1n\!\sum_{j=1}^{n}\E^n\Big[
\left|{\hat \xi}_{n,\varepsilon,M}^{(k)}
- \xi_{n,\varepsilon,M}^{(k)}\right|\!\!(t,z_j)\Big]\!\leq
\frac 1n\!\sum_{j=1}^{n}\!\E^n\Big[\Cal{I}_{\var,t,z_j}^{M} (c)|\eta^{(k)}_{n,\var,M}
-\bar{\eta}^{(k)}|(T_{z_j,c}, c)\Big].
$$

Using lemma \ref{3.2}
and the bound
\begin{equation}
\label{tr}
|\bar{V}^{(k)}
- A_{n,\var, M}^{(k,\{h\},\{j\})}|\leq |\bar{V}^{(k)}
- \hat{A}_{n,\var, M}^{(k)}|+|\hat{A}_{n,\var, M}^{(k)}
- A_{n,\var, M}^{(k,\{h\},\{j\})}|,
\end{equation}
we get, $\forall \de >0$:
$$
\E_{\tau_h}\left[|\eta_{n,\var,M}^{(k)}
- \bar{\eta}^{(k)}|(\cdot,c_h)\right]
\leq  2\de +\frac 1\de \E_{\tau_h}\left[|\Cal{A}^{nM(k)}_{\emptyset,\emptyset}(\{h\}, \{j\})|(\cdot,c_h)\right]
$$
$$
+\frac 1\de \E_{\tau_h}\left[\left\{|\hat{A}^{(k)}_{n,\var,M}
- A_{n,\var, M}^{(k,\{h\},\{j\})}|+
|\hat{A}_{n,\var, M}^{(k)}
- \bar{V}^{(k)}|\right\}
(\cdot,c_h)\right].
$$

Define
$$
\Scr{E}^{(k)}_{h,l}(t)=\Sup_{z\in\R^d\times\R^d}\frac 1{|v|}\E^n
[\Cal{I}_{\var,t,z}^{\{h\}, \{j\}, M}(c_{l})\Cal{A}^{nM(k)}_{\emptyset,\emptyset}(\{h\}, \{j\})(T_{z,c_l},c_l)]
$$

$$
+
\Sup_{1\leq i\leq n}\frac 1{|v_{i}|}\E^n
[\Cal{I}_{\var,t,z_i}^{M}(c_l)
|\bar{V}^{(k)}
- \hat{A}_{n,\var, M}^{(k)}|^2 (T_{z_i,c_l},c_l)]
$$

$$
+
\Sup_{1\leq i\leq n}\frac 1{|v_{i}|}\E^n
[\Cal{I}_{\var,t,z_i}^{\{h\}, \{j\}, M}(c_l)
|A_{n,\var, M}^{(k,\{h\},\{j\})}
- \hat{A}_{n,\var, M}^{(k)}|(T_{z_i,c_l},c_l)].
$$

We get, for all $\de >0$:
$$
\frac 1n\!\sum_{j=1}^{n}\E^n\Big[
\left|{\hat \xi}_{n,\varepsilon,M}^{(k)}
- \xi_{n,\varepsilon,M}^{(k)}\right|(t,z_j)\Big]\leq
\left(\frac1n \Som_{i=1}^{n}|v_i|\right)\left[C_d T \mu \de
+\frac1\de \Sup_{t\in[0,T]}\Scr{E}_{h,h}(t)\right]
$$
and (\ref{particelle}) follows from Lemma \ref{auxil}, Proposition \ref{corrie} and Lemma \ref{gelato}.

As for (\ref{ostacoli}), we use Lemma \ref{3.2} to get the bounds
$$
\E_{\tau_h}\left[|\eta_{n,\varepsilon,M}^{(k)}
- \hat{\eta}_{n,\varepsilon,M}^{(k)}|(t,c_\ell)\right]
\leq  2\de +
$$
$$
\frac 1\de \E_{\tau_h}\left[\left\{|\Cal{A}^{nM(k)}_{\emptyset,\emptyset}(\{h\}, \{j\})|+
|A_{n,\var, M}^{(k,\{h\},\{j\} )}
- \hat{A}^{(k)}_{n,\var,M}|\right\}
(t,c_\ell)\right],
$$
and, for $h\neq \ell$,
$$
\E_{\tau_h}\left[|\eta_{n,\varepsilon,M}^{(k-2)}
- \de^{k-2,\{\ell\},\{j\}}_{n,\varepsilon,M}|(t,c_h)\right]
\leq  2\de +
\frac 1\de \E_{\tau_h}\left[\Cal{A}^{nM(k-2)}_{\emptyset,\emptyset}(\{h\}, \{j\})|(t,c_h)\right],
$$
and we bound the difference of risk functions through (\ref{riduzione}).

Noticing that the quantity $\Cal{A}^{nM(k)}_{\emptyset,\emptyset}(\{\ell\}, \{j\})(t, c_h)$ does not depend on $c_\ell$ for $h\neq\ell$
when $c_\ell \notin \Bigcup_{i=1}^{n} {\Cal T}_\var (t,z_i)$,
we get then, through (\ref{tr}), for all $\de >0$:
$$
\E^n \Big[\var^{d-1} M I_{\mathrm{supp}\phi}(c_\ell)|\eta^{(k)}_{n,\var,M}(t, c_\ell)-\hat{\eta}^{(k)}_{n,\var,M}(t,c_\ell) |\Big]\leq
\mu[\mathrm{supp}\phi]_L K_p (\de+\frac{a_n^d}{\de n})
$$
$$
+\frac{K_d }{\de^3} \left[ \Sup_{s\in[0,T]}\Scr{E}^{(k-2)}_{\ell,\ell-1}(s)+ a_n^{d}\var^{d-1}\Sup_{s\in[0,T]}\Scr{E}^{(k-2)}_{\ell,\ell}(s)
\right.
$$
$$
\left. +\var^{d-1}
\frac1n\Som_{1\leq i \leq n}
\Sup_{\lmul{s\in[0,T]\\z\in\R^d\times\R^d}}\frac 1{|v|}\E^n\!\!\left[\int_{[0,T]} ds_{i} q_n (x_{i} (s_{i}) -c_{\ell}) I_{\{c_{\ell} \in {\Cal T}_\var (s,z)\}}\right.\right.
$$
$$
\left.\left.\times
\Cal{I}_{\var,s_{i}, z_{i}}^{\{\ell\},\{j\}, M}(c_{\ell-1})
M
\Cal{A}^{nM(k-2)}_{\emptyset,\emptyset}(\{\ell\},\{j\})(T_{z_{i},c_{\ell-1}},c_{\ell-1})
\phantom{\int}\right] \right],
$$
where $K_p=K_p(T,\|q\|_{\infty}, K_1, K_2,\mu,C_d)$ and $K_d=K_d(\mu,\Theta,\||v|f_0\|_{L_1})$.
The proof of the proposition follows using Lemmas \ref{auxil}, \ref{gelato} and Proposition \ref{corrie}.
\end{proof}

\bigskip
\section{Final proposition: asymptotic equivalence of (\ref{approxlevk}) and (\ref{sistk})}\label{capitolofinale}
We may now prove our final proposition, which will allow us to establish
the vanishing limits of (\ref{fmu1}) and (\ref{fsi1}) :
\begin{prop}
\label{capbarfin}
For the life functions defined in (\ref{vitak}) and in (\ref{approxlevk}), under the same hypothesis as in Proposition \ref{lem:ind} and  for all $\phi\in C_b(\R^d\times\R^d)$ and $\psi\in C_K(\R^d)$, the following limits are valid, $\Scr{P}$-almost everywhere w.r.t. sequences of initial data $Z_\infty$:
\begin{equation}
\Lim_{n\to\infty}\E^n\Big[|\frac 1n\sum_{i=1}^{n}\phi(T^t(z_i))
({\bar \xi}^{(k)}(t,z_i) - \xi_{n,\varepsilon,M}^{(k)}(t,z_i))|\Big]=0
\end{equation}
and
\begin{equation}
\Lim_{n\to\infty}\E^n \Big[\var^{d-1} \sum_{i=1}^{M}I_{\mathrm{supp}\psi}(c_i)|\eta^{(k)}_{n,\var,M}(t, c_i)-\bar{\eta}^{(k)}(t,c_i) |\Big]=0.
\end{equation}
\end{prop}

\begin{proof}
By the triangular inequality:
$$
\E^n\Big[|\frac 1n\sum_{i=1}^{n}\phi(T^t(z_i))
[{\bar \xi}^{(k)}
- \xi_{n,\varepsilon,M}^{(k)}](t,z_i)
|\Big]\leq
$$
$$
\!\!\!\E^n \Big[|\frac 1n\sum_{i=1}^{n}\phi(T^t(z_i))
[{\bar \xi}^{(k)}
- \hat{\xi}_{n,\varepsilon,M}^{(k)}](t,z_i)
|+\frac{\|\phi\|_{\infty}}n\sum_{i=1}^{n}
|{\hat \xi}^{(k)}_{n,\varepsilon,M}-\xi_{n,\varepsilon,M}^{(k)}|(t,z_i)\Big]
$$
and the right-hand side term vanishes  because of Lemma \ref{lemmauno1} and Proposition \ref{capbar}.

In the same way,
$$
\E^n [\var^{d-1} \sum_{i=1}^{M}I_{\mathrm{supp}\psi}(c_i)|\eta^{(k)}_{n,\var,M}
-\bar{\eta}^{(k)}|(t,c_i)]\leq
$$
$$
\E^n [\var^{d-1} M I_{\mathrm{supp}\psi}(c)\{|\eta^{(k)}_{n,\var,M}
-\hat{\eta}^{(k)}_{n,\var,M}|
+
|\hat{\eta}^{(k)}_{n,\var,M}
-\bar{\eta}^{(k)}|\}(t,c)]
$$
and the right-hand side term vanishes  because of Corollary \ref{equiostcomp} and Proposition \ref{capbar}.
\end{proof}

This proposition completes the proof of theorem \ref{teorema}.

\section{Appendix}
\subsection{Existence and Uniqueness of solutions for the limit system}
\label{app1}
We give here the theorem of existence and uniqueness of solutions of system (\ref{sistlim}), which we can state as follows:

\begin{theo}
\label{esistunic}
Let $\!f_0\!\geq\!0$ and $\sigma_0\!\geq\!0$ be s.t. $\!\!f_0\!\in\!L^1\!(\R_v^d; W^{1,\infty}\!(\R_x^d))$, $\!vf_0\!\in\!L^1\!(\R_v^d; W^{1,\infty}\!(\R_x^d))\cap L^{\infty}(\R_x^d\times \R_v^d)$,  $v^2 f_0\in L^1 (\R_v^d; L^{\infty}(\R_x^d))$ and $\sigma_0\in W^{1,\infty}(\R_x^d)$. Then for each $T>0$ there exists an unique solution $(f,\si)$ to the initial value problem  (\ref{sistlim}) in the interval $[0,T]$.

\end{theo}

The proof of theorem (\ref{esistunic}) is the consequence of the following proposition:

\begin{prop}

\label{mappapf}

Consider the space
$$
\Cal{W}=\{\mathbf{F}=(F_1,F_2,F_3): \, F_i\geq 0\, i=1,2,3,\phantom{F_1\in L^{\infty}([0,T]\times\R^d\times\R^d)}
$$
$$
\phantom{\{\mathbf{F}=(F_1,F_2,F_3):}
F_1\in L^{\infty}([0,T]\times\R^d\times\R^d), F_2,F_3\in L^{\infty}([0,T]\times\R^d)\,\}
$$
with norm
$$
\|\mathbf{F}\|_{\Cal{W}}=\|F_1\|_{L^{\infty}([0,T]\times\R^d\times\R^d)}+\|F_2\|_{L^{\infty}([0,T]\times\R^d)}+\|F_3\|_{L^{\infty}([0,T]\times\R^d)}
$$
and the map $M=(M_1, M_2, M_3):\Cal{W}\to \Cal{W}$ defined, for $f_0,\sigma_0\geq 0$, as:

\begin{equation}
\label{trasf}
\begin{array}{l}
M_1[\mathbf{F}](t,x,v)=f_0(x-vt,v) e^{-C_d |v| \int_0^t ds F_3 (s,x-v(t-s))}\\
\\
M_2[\mathbf{F}](t,x)=\int_{\R^d} dvf_0(x-vt,v) e^{-C_d |v| \int_0^t ds F_3 (s,x-v(t-s))}\\
\\
M_3[\mathbf{F}](t,x)=\si_0 (x) e^{- \Theta \int_0^t ds F_2(s,x)} .
\end{array}
\end{equation}

Assume $\si_0\!\in\!
W^{1,\infty}\!(\R^d\!)$ and $f_0\!\in\!L^1\!(\R_v^d;\!W^{1,\infty}\!(\R_x^d)\!)$,  together with
$vf_0\!\in\!L^1\!(\R_v^d;\!W^{1,\infty}\!(\R_x^d)\!)\cap L^{\infty}(\R_x^d\times \R_v^d)$ and $v^2 f_0\in L^1 (\R_v^d; L^{\infty}(\R_x^d))$.

Then

\begin{itemize}

\item $\!\!M$ is a strict contraction on $\Cal{W}$  when $\!T\!\!<\!T_0$, where $\!T_0$ depends on $C_d$, $\!\Theta$, $\|vf_0\!\|_{\!L^{\infty}\!(\R\times\R^d\!)}$, $\|vf_0\|_{L^{1}(\R_v^d;W^{1,\infty}(\R_x^d))}$, $\|\si\|_{W^{1,\infty}(\R^d)})$

\phantom{ciao}

\item Let $M[\mathbf{f}]=\mathbf{f}$ be the (unique) fixed point of $M$, then for each $t\in [0,T]$, $T< T_0$,
\begin{eqnarray*}
f_1(t,\cdot, \cdot)&\in& L^{1}(\R_v^d;W^{1,\infty}(\R_x^d)),\\
vf_1(t,\cdot, \cdot)&\in& L^1 (\R_v^d; W^{1,\infty}(\R_x^d))\cap L^{\infty}(\R_x^d\times \R_v^d)
,\\
v^2 f_1(t,\cdot, \cdot)&\in& L^1 (\R_v^d; L^{\infty}(\R_x^d))\\
f_3(t,\cdot)&\in& W^{1,\infty}(\R^d)
.
\end{eqnarray*}

\end{itemize}

\end{prop}

\begin{rem}
The map (\ref{trasf}) is defined so to represent the solutions to the linear problem (\ref{sistk}) (cf. the proof of Theorem \ref{esistunic} on next page), and therefore it does not depend on $F_1$ (the sources in the linear problem are indeed $F_2$ and $F_3$). Three components are (of course) nevertheless needed to get a map having as unique fixed point the solution to the nonlinear problem (\ref{sistlim}).
\end{rem}

\begin{proof}[Proof of Proposition \ref{mappapf}]

Consider $\mathbf{F}, \mathbf{G}\in \Cal{W}$. Then
$$
\begin{array}{l}
\!\!\!\|M_1[\mathbf{F}]\!-\!M_1[\mathbf{G}]\|_{L^{\infty}([0,T]\!\times\!\R^d\!\times\!\R^d)}\!\leq \!
C_d T \|vf_0\|_{L^{\infty}(\R^d\!\times\!\R^d)} \|F_3\!-\!G_3\|_{L^{\infty}([0,T]\!\times\!\R^d)}\\
\\
\!\!\!\|M_2[\mathbf{F}]\!-\!M_2[\mathbf{G}]\|_{L^{\infty}([0,T]\!\times\!\R^d)}\!\leq\!
C_d T \|vf_0\|_{L^{1}(\R^d;L^{\infty}(\R^d))}  \|F_3\!-\!G_3\|_{L^{\infty}([0,T]\!\times\!\R^d)}\\
\\
\!\!\!\|M_3[\mathbf{F}]\!-\!M_3[\mathbf{G}]\|_{L^{\infty}([0,T]\!\times\!\R^d)}\!\leq\!
\Theta T \|\si_0\|_{L^{\infty}([0,T]\!\times\!\R^d)}\|F_2\!-\!G_2\|_{L^{\infty}([0,T]\!\times\!\R^d)}
\end{array}
$$

We can therefore write
\begin{equation}
\label{contract}
\|M[\mathbf{F}]-M[\mathbf{G}]\|_{\Cal{W}} \leq A T \|\mathbf{F}- \mathbf{G}\|_{\Cal{W}},
\end{equation}
where $\!A$ is a constant depending on $C_d$, $\!\Theta$ and on the norms $\|vf_0\|_{\!L^{\infty}\!(\R^d\!\times\R^d\!)}$, $\|\si_0\|_{\!L^{\infty}\!([0,T]\!\times\R^d\!)}$, $\|vf_0\|_{L^{1}(\R^d_v; L^{\infty}(\R_x^d))}$. Whenever $T<T_0$, where $T_0=\frac 1A$, $M$ is a contraction on the complete space $\Cal{W}$. Therefore, whenever $T< \frac 1A$, there exists a unique fixed point $(f,g,\si)=M(f,g,\si)$ (and of course, $g=\int dv f$).

The properties of the fixed point functions listed in the second part of the thesis of the proposition follow trivially from (\ref{trasf}) and the assumptions on $f_0$.
\end{proof}

\begin{proof}[Proof of theorem \ref{esistunic}]

The map $M$ maps the point $\mathbf{F}$ in the point $M[\mathbf{F}]$, with $M_2[\mathbf{F}]=\int\!dv M_1\![\mathbf{F}]$ and  $(\!M_1\![\mathbf{F}],\!M_3\![\mathbf{F}]\!)$ solution to the linear problem (\ref{sistk}), with sources $\int_{\R^d}\!dv f^{(k-1)}\!=F_2$ and $\si^{(k-1)}=F_3$ and initial data $f(0,x,v)=f_0(x,v)$, $\si(0,x)=\si_0(x)$. The fixed point $\mathbf{f}$ of $M$ is therefore s.t. $\mathbf{f}=(f,\int dv f , \si)$, where $(f,\si)$ is the solution of (\ref{sistlim}) for $t\in [0,T]$, with $T< \frac 1A$.

Because of the property of the fixed point $(f,\int dv f,\si)$, the solution is prolongeable for any value of $T>0$.
\end{proof}

\smallskip

\begin{rem}
Since $f_1(t,\cdot, \cdot) \in L^{1}(\R_v^d;W^{1,\infty}(\R_x^d))$, $\int dv f_1(t,\cdot,v)\in C_b(\R^d)$.
\end{rem}

\subsection{Condition for the weak convergence $q_n \mu_0^n\deb\delta_0 f_0$}\label{app2}
We state here a simple condition on the growth rate of the generic term of the sequence $\{a_n\}$ such that the product of two weakly convergent measures converges weakly to the product of the two limit measures.

\begin{lemma}
\label{convergdebole}

Let $\Scr{P}$ be a probability measure defined as (\ref{pgrande}), with one particle probability density $f_0$ s.t.
\begin{equation}
\label{velo}
\int dv f_0(\cdot, v) \in \mathscr{S}(\R^d),
\end{equation}
and $q\in \mathscr{S}(\R^d)$ a non negative function s.t.
$$
\int q(x)dx = \Theta >0.
$$
Take a sequence of positive real numbers $\{a_n\}_{n=1}^{\infty}$ ($a_n>0$) such that for some $\kappa \in (0,\frac 12)$
$$
\Lim_{n\to\infty} \frac{a_n^d}{n^{\kappa}}=0.
$$

Then, given $\phi\in C_b(\R^d\times \R^d)$,  $\Scr{P}$-a.s.,
\begin{equation}
\frac 1n \sum_{h=1}^{n} a_n^d q(a_n x_h)\phi(x_h, v_h) \to \Theta \int_{\R^d} f_0(0,v) \phi(0,v) dv .
\end{equation}

\end{lemma}

\begin{proof}

We observe first that, since $q\in \mathscr{S}(\R^d)$, we can use the identity
$$
q(x) =\frac{1}{(2\pi)^d}\int_{\R^d} e^{i k\cdot x} \hat{q}(k)dk
$$
where $\hat{q}\in \mathscr{S}(\R^d)$ is the Fourier transform of $q$.

Because of $\hat{q}\in \mathscr{S}(\R^d)$ and (\ref{velo}), in all calculations below  we can apply Fubini's theorem.

We can then write:

\begin{equation}
\begin{split}
\left|\frac 1n \sum_{h=1}^{n} a_n^d q(a_n x_h)\phi(x_h, v_h) -\Theta \int_{\R^d} f_0(0,v) \phi(0,v) dv\right|\leq \\
\\
\frac{1}{(2\pi)^d}\!\!\int_{\R^d} \big|\hat{q}\big(\frac{k}{a_n}\big)\big|
\left|\frac 1n \sum_{h=1}^{n}e^{i k\cdot x_h}\phi(x_h, v_h) -\!\!\int_{\R^d\times\R^d} \!\!\!\!\!\!e^{i k\cdot x} f_0(x,v) \phi(x,v)dx dv \right|dk\\
\\
+ \left|\int_{\R^d\times\R^d} a_n^d q(a_n x) f_0(x,v) \phi (x,v)dx dv -\Theta \int_{\R^d} f_0(0,v) \phi(0,v) dv \right| .
\end{split}
\end{equation}

We consider then, for $h=1,\ldots$, the sequence of independent stochastic variables
$$
\nu^{\phi}_h (k) = e^{i k\cdot x_h} \phi(x_h,v_h) - \int_{\R^d\times\R^d} e^{i k\cdot x}f_0(x,v) \phi(x,v),
$$
s.t. $\E_{\Scr{P}} [\nu^{\phi}_h(k)]=0$ and $\E_{\Scr{P}}[|\nu^{\phi}_h|^2 (k)]\leq 2 \|\phi\|_{\infty}^{2}$
and for which we have, for any even $j\in \N$,
\begin{equation}
\label{mompari}
\E_{\Scr{P}} \left[|\frac 1n \sum_{h=1}^{n}\nu^{\phi}_h (k)|^j \right]\leq \frac{K_j}{n^{\frac j2}}
(\E_{\Scr{P}}[|\nu|^2(k)])^{\frac j2} \leq \frac{K_j}{n^{\frac j2}}
(\sqrt{2} \|\phi\|_{\infty})^j.
\end{equation}

We define then
\begin{equation}
\hat{\omega}^{\phi}_n=
\frac{1}{(2\pi)^d}\int_{\R^d}\!\!dk \big|\hat{q}\big(\frac{k}{a_n}\big)\big|
\left|\frac 1n \sum_{h=1}^{n}\nu^{\phi}_h (k)\right| I_{\{k: |\frac 1n \Som_{h=1}^{n}\nu^{\phi}_h (k)|> \frac{C}{n^{\kappa}}\}}
\end{equation}
and we get:
\begin{equation}
\begin{split}
\left|\frac 1n \sum_{h=1}^{n} a_n^d q(a_n x_h)\phi(x_h, v_h) -\Theta \int_{\R^d} f_0(0,v) \phi(0,v) dv\right|\leq
C \frac{a_n^d}{n^{\kappa}}\|\hat{q}\|_{L^1} +  \hat{\omega}^{\phi}_n\\
\\
+ \left|\int_{\R^d\times\R^d} a_n^d q(a_n x) f_0(x,v) \phi (x,v)dx dv -\Theta \int_{\R^d} f_0(0,v) \phi(0,v) dv \right|
\end{split}
\end{equation}

Since, because of the characteristic function in the definition of $\hat{\omega}^{\phi}_n$,
$$
\E_{\Scr{P}}(\hat{\omega}^{\phi}_n) \leq  \frac{n^{(j-1)\kappa}}{C (2\pi)^d}\int_{\R^d}\!\!dk \big|\hat{q}\big(\frac{k}{a_n}\big)\big| \E_{\Scr{P}} \left[|\frac 1n \sum_{h=1}^{n}\nu^{\phi}_h (k)|^j \right],
$$

by Tchebycheff inequality and (\ref{mompari}),  for $\de>0$,
$$
\Scr{P}(|\hat{\omega}^{\phi}_n|> \de) \leq \frac{\E_{\Scr{P}}(\hat{\omega}^{\phi}_n)}{\de}\leq \frac{K}{\de} a^d_n n^{(j-1)\kappa - \frac j2}
$$
(with $K$ depending on $\|\phi\|_{\infty}$ and $\|\hat{q}\|_{L^1}$); we obtain then, for $\kappa\in(0,\frac 12 )$ and for all $\de>0$, $\Som_n \Scr{P}(|\hat{\omega}^{\phi}_n|> \de) <\infty$, and therefore
$\hat{\omega}^{\phi}_n \stackrel{\Scr{P} a.s.}{\to} 0$.

Since $\Lim_{n\to\infty}\frac{a_n^d}{n^{\kappa}}=0$ and $q_n(x)=a_n^d q(a_n x)$ is such that $q_n \deb \Theta \delta_0$, we get finally
$$
\frac 1n \sum_{h=1}^{n} a_n^d q(a_n x_h)\phi(x_h, v_h) \stackrel{\Scr{P} a.s.}{\to} \Theta \int_{\R^d} f_0(0,v) \phi(0,v) dv .
$$
\end{proof}

\begin{prop}
\label{cdebole}
In the same hypothesis as in Lemma \ref{convergdebole}, on a full measure set with respect to $\Scr{P}$
$$
q_n(x) \mu_0^n(x,v) \deb \Theta\delta_0(x) f_0(x,v).
$$
\end{prop}

\begin{proof}
We consider first functions in the separable space $C_0(\R^d\times\R^d)$, the set of continuous functions vanishing at infinity. Let $D$ be a countable dense set in  $C_0(\R^d\times\R^d)$ and consider the two sets
$$
\!\!\!\!A=\{ Z_{\infty}: \forall \phi \in D \,\int_{\R^d\times\R^d}\!\!\!\!dxdv q_n(x) \mu_0^n (x,v) \phi(x,v) \to \Theta \int_{\R^d}\!\!\! dv f_0(0,v) \phi(0,v)\}
$$
and
$$
B=\{ Z_{\infty}: \int_{\R^d\times\R^d}dxdv q_n(x) \mu_0^n (x,v) \to \Theta\int_{\R^d} dv f_0(0,v)\}.
$$

Because of Lemma \ref{convergdebole}, both sets have full measure, so as their intersection $A\cap B$.

For $Z_{\infty}\in A\cap B$ the sequence $\{q_n \mu_n^0\}_{n=1}^{\infty}$ is a sequence of finite  positive measures s.t. for all $\phi\in C_0(\R^d\times\R^d)$
\begin{equation}
\label{ciccio}
\int_{\R^d\times\R^d} dx dv q_n \mu_0^n \phi \to \Theta\int_{\R^d} dv
f_0(0,v)\phi(0,v)< \infty.
\end{equation}
Since the convergence in (\ref{ciccio}) is valid also for $\phi=1$, on the set $A\cap B$, weakly in the sense of measure,
$$
q_n (x)\mu_n (x,v) \deb \Theta\delta_0 (x) f_0 (x,v)
$$
(see e.g. \cite{MA}, p.90, theorem 6.8), and the proposition is proved.
\end{proof}

\begin{rem}
\label{rdebomom}
Since $f_0\in \Scr{S}(\R^d\times\R^d)$ and $\int dv f_0 \in \Scr{S}(\R^d)$, the convergence proved in Proposition \ref{cdebole} is valid also for $|v|^j q_n \mu^n_0$, $j=1,2,\ldots$ (i.e. $|v|^j q_n \mu^n_0\deb |v|^j \delta_0 f_0$). To show it, it suffices to rewrite the proof with $\nu_h^{\phi}$ replaced with $|v|^j\nu_h^{\phi}$. Under the same hypothesis, it is also possible to prove, by induction, that $\otimes_{k=1}^{M} q_n\mu_n \deb \otimes_{k=1}^M \de_0 f_0 $, for $M=1,2\ldots$.
\end{rem}

\begin{rem}
\label{rdebotrasl}
Whenever, in addition to the hypothesis in Lemma \ref{convergdebole}, we have $f_0\in L^1(\R_v^d; W^{1,\infty}(\R_x^d))$ and $v^2 f_0\in L^1 (\R_v^d; L^{\infty}(\R_x^d))$, we can obtain, following the same steps, the weak convergence of $(\Cal{T}^{v,s}_a q)_n \mu_n^0$ and $|v|(\Cal{T}^{v,s}_a q)_n \mu_n^0$, where we define $(\Cal{T}^{v,s}_a g)(x)=g(x+vt+a)$. The convergence is uniform in $a \in\R^d$, as can be easily checked.

Moreover, a very simplified form of the procedure allows to prove that, given a limit density  $f_0\in \Scr{S}(\R^d\times\R^d)$ and two positive integers $P, Q$, on a full measure set w.r.t. $\Scr{P}$, $|v|^j \mu^n_0\deb |v|^j f_0$  and $\otimes_{k=1}^{i} |v_k|^j\mu_n \deb \otimes_{k=1}^i |v|^j f_0 $, for $j=0,\ldots,P$ and $i=1,\ldots, Q$.
\end{rem}


\medskip
Received xxxx 20xx; revised xxxx 20xx.
\medskip

\end{document}